\documentclass[12pt]{article}
\usepackage{amssymb,bm}
\usepackage{amsmath}
\usepackage{amsthm}
\usepackage{graphicx}
\usepackage{hyperref} 
\hypersetup{pdfborder=0 0 0}
\newtheorem{theorem}{{\sc Theorem}}[section]
\newtheorem{proposition}[theorem]{{\sc Proposition}}

\newtheorem{lemma}[theorem]{{\sc Lemma}}
\newtheorem{corollary}[theorem]{Corollary}

\newtheorem{definition}[theorem]{Definition}

\newcommand{\bb}[1]{\mathbb{ #1}}

\bmdefine\Bone{1}

\newcommand{\weak}{\rightharpoonup\:}

\newcommand{\bra}[1]{\overline{#1}}

\newcommand{\hf}{\displaystyle\frac{1}{2}}
\newcommand{\nth}[1]{\displaystyle\frac{1}{#1}}

\renewcommand{\Hat}[1]{\widehat{#1}}
\newcommand{\Tld}[1]{\widetilde{#1}}

\newcommand{\av}[1]{\langle #1 \rangle}

\def\XXint#1#2#3{{\setbox0=\hbox{$#1{#2#3}{\int}$ }
\vcenter{\hbox{$#2#3$ }}\kern-.6\wd0}}

\newcommand{\im}{\mathfrak{Im}}
\newcommand{\re}{\Re\mathfrak{e}}

\newcommand{\lims}{\mathop{\overline\lim}}
\newcommand{\limi}{\mathop{\underline\lim}}

\newcommand{\rhs}{right-hand side}
\newcommand{\lhs}{left-hand side}

\newcommand{\nbh}{neighborhood}
\newcommand{\IFF}{if and only if }

\newcommand{\Ga}{\alpha}
\newcommand{\Gb}{\beta}
\newcommand{\Gd}{\delta}
\newcommand{\Ge}{\epsilon}

\newcommand{\Gve}{\varepsilon}

\newcommand{\Gg}{\gamma}

\newcommand{\Gk}{\kappa}

\newcommand{\Gl}{\lambda}

\newcommand{\Gth}{\theta}

\newcommand{\Gs}{\sigma}

\newcommand{\Go}{\omega}

\newcommand{\GD}{\Delta}

\newcommand{\GG}{\Gamma}

\bmdefine\BGa{\alpha}
\bmdefine\BGb{\beta}
\bmdefine\BGd{\delta}
\bmdefine\BGe{\epsilon}
\bmdefine\BGve{\varepsilon}
\bmdefine\BGf{\phi}
\bmdefine\BGvf{\varphi}
\bmdefine\BGg{\gamma}
\bmdefine\BGc{\chi}
\bmdefine\BGi{\iota}
\bmdefine\BGk{\kappa}
\bmdefine\BGl{\lambda}
\bmdefine\BGn{\eta}
\bmdefine\BGm{\mu}
\bmdefine\BGv{\nu}
\bmdefine\BGp{\pi}
\bmdefine\BGth{\theta}
\bmdefine\BGvth{\vartheta}
\bmdefine\BGr{\rho}
\bmdefine\BGvr{\varrho}
\bmdefine\BGs{\sigma}
\bmdefine\BGvs{\varsigma}
\bmdefine\BGt{\tau}
\bmdefine\BGj{\tau}
\bmdefine\BGu{\upsilon}
\bmdefine\BGo{\omega}
\bmdefine\BGx{\xi}
\bmdefine\BGy{\psi}
\bmdefine\BGz{\zeta}
\bmdefine\BGD{\Delta}
\bmdefine\BGF{\Phi}
\bmdefine\BGG{\Gamma}
\bmdefine\BGL{\Lambda}
\bmdefine\BGP{\Pi}
\bmdefine\BGT{\Theta}
\bmdefine\BGS{\Sigma}
\bmdefine\BGU{\Upsilon}
\bmdefine\BGO{\Omega}
\bmdefine\BGX{\Xi}
\bmdefine\BGY{\Psi}

\bmdefine\BFM{\mathfrak{M}}
\bmdefine\BFb{\mathfrak{b}}
\bmdefine\BFk{\mathfrak{k}}
\bmdefine\BFm{\mathfrak{m}}
\bmdefine\BFu{\mathfrak{u}}
\bmdefine\BFv{\mathfrak{v}}

\newcommand{\CB}{{\mathcal B}}

\newcommand{\CK}{{\mathcal K}}
\newcommand{\CL}{{\mathcal L}}

\newcommand{\CR}{{\mathcal R}}
\newcommand{\CS}{{\mathcal S}}

\bmdefine\BCA{{\mathcal A}}
\bmdefine\BCB{{\mathcal B}}
\bmdefine\BCC{{\mathcal C}}
\bmdefine\BCD{{\mathcal D}}
\bmdefine\BCE{{\mathcal E}}
\bmdefine\BCF{{\mathcal F}}
\bmdefine\BCG{{\mathcal G}}
\bmdefine\BCH{{\mathcal H}}
\bmdefine\BCI{{\mathcal I}}
\bmdefine\BCJ{{\mathcal J}}
\bmdefine\BCK{{\mathcal K}}
\bmdefine\BCL{{\mathcal L}}
\bmdefine\BCM{{\mathcal M}}
\bmdefine\BCN{{\mathcal N}}
\bmdefine\BCO{{\mathcal O}}
\bmdefine\BCP{{\mathcal P}}
\bmdefine\BCQ{{\mathcal Q}}
\bmdefine\BCR{{\mathcal R}}
\bmdefine\BCS{{\mathcal S}}
\bmdefine\BCT{{\mathcal T}}
\bmdefine\BCU{{\mathcal U}}
\bmdefine\BCV{{\mathcal V}}
\bmdefine\BCW{{\mathcal W}}
\bmdefine\BCX{{\mathcal X}}
\bmdefine\BCY{{\mathcal Y}}
\bmdefine\BCZ{{\mathcal Z}}

\bmdefine\Bzr{ 0}
\bmdefine\Ba{ a}
\bmdefine\Bb{ b}
\bmdefine\Bc{ c}
\bmdefine\Bd{ d}
\bmdefine\Be{ e}
\bmdefine\Bf{ f}
\bmdefine\Bg{ g}
\bmdefine\Bh{ h}
\bmdefine\Bi{ i}
\bmdefine\Bj{ j}
\bmdefine\Bk{ k}
\bmdefine\Bl{ l}
\bmdefine\Bm{ m}
\bmdefine\Bn{ n}
\bmdefine\Bo{ o}
\bmdefine\Bp{ p}
\bmdefine\Bq{ q}
\bmdefine\Br{ r}
\bmdefine\Bs{ s}
\bmdefine\Bt{ t}
\bmdefine\Bu{ u}
\bmdefine\Bv{ v}
\bmdefine\Bw{ w}
\bmdefine\Bx{ x}
\bmdefine\By{ y}
\bmdefine\Bz{ z}
\bmdefine\BA{ A}
\bmdefine\BB{ B}
\bmdefine\BC{ C}
\bmdefine\BD{ D}
\bmdefine\BE{ E}
\bmdefine\BF{ F}
\bmdefine\BG{ G}
\bmdefine\BH{ H}
\bmdefine\BI{ I}
\bmdefine\BJ{ J}
\bmdefine\BK{ K}
\bmdefine\BL{ L}
\bmdefine\BM{ M}
\bmdefine\BN{ N}
\bmdefine\BO{ O}
\bmdefine\BP{ P}
\bmdefine\BQ{ Q}
\bmdefine\BR{ R}
\bmdefine\BS{ S}
\bmdefine\BT{ T}
\bmdefine\BU{ U}
\bmdefine\BV{ V}
\bmdefine\BW{ W}
\bmdefine\BX{ X}
\bmdefine\BY{ Y}
\bmdefine\BZ{ Z}

\DeclareMathAlphabet{\cg}{OMS}{zplm}{m}{n}
\newcommand{\HH}{\mathbb{H}}
 
\newcommand{\RR}{\mathbb{R}}      
\newcommand{\ZZ}{\mathbb{Z}}      

\usepackage{calrsfs}
\usepackage{subcaption}
\DeclareMathAlphabet{\pazocal}{OMS}{zplm}{m}{n}
\newcommand{\PK}{\pazocal{K}}
\newcommand{\PC}{\pazocal{C}}

\title{On feasibility of extrapolation of the complex electromagnetic
  permittivity function using Kramer-Kronig relations} 
\author{Yury Grabovsky,  \qquad Narek Hovsepyan}
\date{}
\begin{document}
\maketitle

\begin{abstract}
We study the degree of reliability of extrapolation of complex electromagnetic permittivity functions based on their analyticity properties. Given two analytic functions, representing extrapolants of the same experimental data, we examine how much they can differ at an extrapolation point outside of the experimentally accessible frequency band. We give a sharp upper bound on the worst case extrapolation error, in terms of a solution of an integral equation of Fredholm type. We conjecture and give numerical evidence that this bound exhibits a power law precision deterioration as one moves further away from the frequency band containing measurement data.
\end{abstract}
\tableofcontents

\section{Introduction}
\setcounter{equation}{0} 
\label{SECT intro} 
Properties of linear, time-invariant, causal systems are characterized by
functions analytic in a complex half-plane. Examples include transfer
functions of digital filters \cite{grs01}, complex impedance and admittance
functions of electrical circuits \cite{brune31}, complex magnetic permeability
and complex dielectric permittivity functions \cite{lali60:8,feyn64}. Arising
from the world of real-valued fields, these functions also possess specific
symmetries. The underlying mathematical structure is the Fourier (or Laplace)
transforms of real-valued functions that vanish on negative semi-axis. More
generally, the analyticity arises from the analyticity of resolvents of linear
operators, while their symmetries reflect that these operators are very often
real and self-adjoint.

In a typical situation we can measure the values of such analytic functions on
a compact subset of the boundary of their half-plane of analyticity. The real
and imaginary parts of such a function are not independent but are Hilbert
transforms of one another. In the context of the complex dielectric
permittivity this fact is expressed by the Kramers-Kronig relations
\cite{kron26}. It is therefore tempting to use these relations in order to
reconstruct the analytic functions from their measured values. Unfortunately,
such a reconstruction problem is ill-posed (e.g. \cite{mill70}), and one needs
to place additional constraints on the set of admissible analytic functions
for the extrapolation problem to be mathematically well-posed.

In this paper we propose a physically natural regularization that implies that
the underlying analytic functions can be analytically continued into a larger
complex half-plane. In that case, the idea is to exploit the fact that complex
analytic functions possess a large degree of rigidity, being uniquely
determined by values at any infinite set of points in any finite
interval. This rigidity also implies that even very small measurement errors
will produce data \textit{mathematically} inconsistent with values of an
analytic function. In such cases the least squares approach
\cite{cude68,ciulli69,capr74,capr79} that treats all data points equally is
the most natural one. In the first part of the paper we prove that the least
squares problem has a unique solution, that yields a mathematically stable
extrapolant. We show that the minimizer must be a rational function and derive
the necessary and sufficient conditions for its optimality.

Recent work \cite{trefe19,deto18,grho-gen,grho-annulus} shows that
surprisingly, the space of analytic functions is also "flexible" in the sense
that the data can often be matched up to a given precision by two physically
admissible functions that are very different away from the interval, where the
data is available. The second part of the paper quantifies this phenomenon by
giving an optimal upper bound on the possible discrepancy between any two
approximate extrapolants. This is done by first reformulating the problem as a
question about analytic functions, which we have already studied in
\cite{grho-gen,grho-annulus}, but without the symmetry
constraints. Incorporating symmetry into the methods of \cite{grho-gen} is
nontrivial, and we address this question next. Our conclusion
is that the symmetry has a virtually negligible regularizing effect, as far as the
optimal upper bound on the extrapolation uncertainty is concerned.

\section{Preliminaries}
\setcounter{equation}{0} 
\label{SECT prelim} 
When the electromagnetic wave passes
through the material the incident electric field $\BE(\Bx,t)$ interacts with
charge carriers inside the matter. We assume that the induced polarization
field $\BP(\Bx,t)$ depends on the incident electric field linearly and
locally. This is expressed by the constitutive relation 
\begin{equation}
  \label{constrel}
  \BP(\Bx,t)=\int_{0}^{+\infty}\BE(\Bx,t-s)a(s)ds,
\end{equation}
indicating that the polarization field depends only on the past values of
$\BE(\Bx,t)$. The function $a(t)$ is called the impulse response or a memory
kernel, which is assumed to decay exponentially. Its decay rate, $a(t)\sim
e^{-t/\tau_{0}}$, $t\to\infty$, indicates how fast the system ``forgets'' the past
values of the incident field. The parameter $\tau_{0}>0$ is called the relaxation
time, which can be measured for many materials.

Let 
\[
a_{0}(t)=
\begin{cases}
  a(t),&t\ge 0,\\
  0,&t<0.
\end{cases}
\]
Then we can extend the integral in (\ref{constrel}) to the entire real line
and apply the Fourier transform to convert the convolution into a product:
\[
\Hat{\BP}(\Bx,\omega)=\Hat{a}_{0}(\omega)\Hat{\BE}(\Bx,\omega),
\]
where 
\[
\Hat{f}(\omega)=\int_{\bb{R}}f(x)e^{i \omega x}dx 
\] 
is the Fourier transform. In physics, the function $\Gve(\omega)=\Gve_{0}+\Hat{a}_{0}(\omega)$
is called the complex dielectric permittivity of the material, where
$\Gve_{0}$ is the dielectric permittivity of the vacuum. Mathematically, it is
more convenient to study $\Hat{a}_{0}(\omega)$, rather than $\Gve(\omega)$. From
now on, we will denote
\[
f(\omega)=\Hat{a}_{0}(\omega),
\]
and refer to it as the complex electromagnetic permittivity, in a convenient
abuse of terminology. Let us recall the well-known analytic properties of
isotropic complex electromagnetic permittivity as a function of frequency
$\omega$ of the incident electromagnetic wave \cite{lali60:8,feyn64}:
\begin{enumerate}
\item[(a)] $\overline{f(\omega)} = f(- \overline{\omega})$;

\item[(b)] $f(\omega)$ is analytic in the complex upper half-plane
  $\HH_{+}=\{\Go\in\bb{C}:\im\,\Go>0\}$ 
\item[(c)] $\im\, f(\omega) > 0$ for $\Go$ in the first quadrant
  $\re(\omega)>0$, $\im(\omega)>0$;
\item[(d)] $f(\omega) = -A \omega^{-2} + O(\omega^{-3})$, $A>0$ as $\omega \rightarrow \infty$.
\end{enumerate}
Property (a) expresses the fact that physical fields are real. Property (b) is
the consequence of the causality principle i.e. independence of $\BP(\Bx,t)$
of the future values of $E(\Bx,\tau)$, $\tau>t$. Property (c) comes from the
fact that the electromagnetic energy gets absorbed by the material as the
electromagnetic wave passes through. Property (d) is called the plasma limit,
where at very high frequencies the electrons in the medium may be regarded as
free.  Complex analytic functions with properties (a)--(d) and their variants,
are ubiquitous in physics.  The complex impedance of electrical circuits as a
function of frequency has similar properties \cite{fost24,brune31}. Yet another example is the
dependence of effective moduli of composites on the moduli of its constituents
\cite{Bergman:1978:DCC,mi80,mchcg15}. These functions appear in areas as diverse as
optimal design problems \cite{lipt01} and nuclear physics
\cite{mand58,macd59,cfg68}. Typically\footnote{In the context of viscoelastic
  composites measurements corresponding to values of $f(\Go)$ in the upper
  half-plane are also possible.} only the values of such a function on a real
line can be measured. In the case of complex electromagnetic permittivity the
measurements are usually made either on a finite interval or at a discrete set
of frequencies. However, the requirements (a)--(d) do not place any regularity
requirements on $f(\Go)$, when $\Go$ is real. For example, the function
\[
f(\Go)=\nth{\Go_{0}^{2}-\Go^{2}},\quad\Go_{0}>0
\]
satisfies properties (a)--(d), but blows up at the frequency $\Go_{0}>0$. We
exclude such examples by assuming that the memory kernel $a(t)$ decays
exponentially with relaxation time $\tau_{0}>0$. In this case $f(\Go)$ will
have an analytic extension into the larger half-plane
\begin{equation}
\label{Hhdef}
\HH_h = \{\omega \in \mathbb{C}: \im\, \omega> -h\},
\end{equation}
where $h=1/\tau_{0} >0$. In general, the analytic continuation of $f(\Go)$
need not have positive imaginary part when $\im(\Go)>-h$ and $\re(\Go)>0$.
For example, $f(\omega) = - \frac{\omega+i}{(\omega+3i)^3}$ satisfies
conditions (a)--(d), is analytic in $\HH_{3}$, but $\im\, f(x-i\Ge)$ takes
negative values for any $\Ge \in (0,3)$ for some $x>0$. We therefore make an
additional regularizing assumption that positivity property (c) continues to
hold in the larger half-plane $\HH_h$. In fact, under the additional
assumption that the Elmore delay \cite{elmo48} is positive, i.e., $-if'(0)>0$, the
positivity condition can be guaranteed in some possibly smaller half-plane
$\HH_{h'}$, $0<h'\le h$ (see the Appendix). Thus, the class of
all physically admissible complex dielectric permittivity functions is narrowed in
a natural way to the class $\PK_{h}$, defined as follows.
\begin{definition}
  \label{def:Kh}
A complex analytic function $f:\HH_{h}\to\bb{C}$ belongs to the class
$\PK_{h}$ if it has the following list of physically justified properties.
\begin{enumerate}
\item[(S)] Symmetry: $\overline{f(\omega)} = f(- \overline{\omega})$;
\item[(P)] Passivity: $\im(f(\Go))>0$, when $\im(\Go)>-h,\quad\re(\Go)>0$;
\item[(L)] Plasma limit: $f(\omega) = -A \omega^{-2} + O(\omega^{-3})$, $A>0$ as $\omega \rightarrow \infty$.
\end{enumerate}
\end{definition}
Functions in the set $\PK_{h}$ are closely related to an important class of
functions called Stieltjes functions.
\begin{definition}
  \label{def:Stiel}
  A non-constant function analytic in the complex upper half-plane is said to
  be of Stieltjes class $\mathfrak{S}$ if its imaginary part is positive, and
  it is analytic on the negative real axis, where it takes real and nonnegative
  values. Such functions together with all nonnegative constant functions form
  the Stieltjes class $\mathfrak{S}$.
\end{definition}
It is well-known that a Stieltjes function $F(z)$ is uniquely determined by a
constant $\rho \ge 0$ and a Borel-regular positive measure $\Gs$ by the
representation
\begin{equation}
  \label{Strep}
  F(z)=\rho + \int_{0}^{\infty}\frac{d\Gs(\Gl)}{\Gl-z},\qquad
\int_{0}^{\infty}\frac{d\Gs(\Gl)}{\Gl+1}<+\infty.
\end{equation}
The measure $\Gs$ is often referred to as \emph{spectral measure}
\cite{kron26,lali60:8}. Let us show that function $f\in\PK_{h}$ can be represented by
\begin{equation}
  \label{intrep}
f(\Go)=F((\Go+ih)^{2}),\quad F\in\mathfrak{S},\quad\rho=0,\quad
\int_{0}^{\infty}d\Gs(\Gl)=A<+\infty,
\end{equation}
where  $\Gs$ is the spectral measure for $F(z)$.

For any $f \in \PK_h$ consider the function $g(\zeta) = f(\zeta-ih)$ which is analytic in $\HH_+$, \ $\overline{g(\zeta)} = g(- \overline{\zeta})$, \ $\im g >0$ in the first quadrant and $g(\zeta) \sim -A \zeta^{-2}$ as $\zeta \to \infty$ for some $A>0$.

Unfolding the first quadrant in the $\zeta$-plane into the upper half-plane in
the $z$-plane via $z=\zeta^2$ we obtain a function $F(z) = g(\sqrt{z})$, which
is analytic in $\HH_+$ and has a positive imaginary part there. The symmetry
of $g$ implies that it is real on $i\RR_{>0}$, but then $F$ is real on
$\RR_{<0}$. Clearly, analyticity of $g$ on $i\RR_{>0}$ implies that of $F$ on
$\RR_{<0}$. The plasma limit assumption implies that $F(-x) \geq 0$ for $x$ large enough, which is enough to conclude that $F$ is a Stieltjes function (see the proof of Theorem A.4, pg 392 of \cite{krnu77}). Thus, $F$ admits the
representation (\ref{Strep}). But then, the asymptotic relation $F(z) \sim -A
z^{-1}$ as $z \to \infty$ implies that $\rho = 0$ and $\int_0^\infty d \sigma
(\lambda) = A < \infty$. Thus, $f(\omega) = g(\omega + ih) = F((\omega +
ih)^2)$. Conversely, if $f$ is given by \eqref{intrep} then it is straightforward to check that it satisfies all the required properties of class $\PK_h$.

\section{Main results} \label{SECT main results}

Let us assume that the experimentally measured data $f_{\rm exp}(\Go)$ is known on
a band of frequencies $\Gamma = [0,B]$. The unavoidable random noise makes the
measured values mathematically inconsistent with the analyticity of the
complex dielectric permittivity function. The standard way to deal with the
noise is to use the ``least squares'' approach by looking for a function
$f\in\PK_{h}$ that is closest to the experimental data $f_{\rm exp}(\Go)$ in the
$L^{2}$ norm on $\GG$. Thus, after rescaling the frequency interval $\GG$ to
the interval $[0,1]$ we arrive at the following least squares problem
\begin{equation}
\label{lsq}
\inf_{f\in \PK_h} \|f-f_{\rm exp}\|_{L^2(0,1)}.
\end{equation}
One approach \cite{echerk01,chou08} is to ignore the positivity requirement,
while retaining the spectral representation (\ref{intrep}). The resulting
problem constrains $f$ to a vector space, but becomes ill-posed. It is then
solved by Tikhonov regularization techniques. Unfortunately, such an approach
cannot guarantee that the solution possesses the required positivity.

We will see in Section~\ref{SECT lsq} that the positivity property of
functions in $\PK_{h}$ plays a regularizing role, making the least squares
problem (\ref{lsq}) well-posed. So the solution to \eqref{lsq} exists, is
unique and lies in the closure $\CS_h = \overline{\PK_h}$ with respect to the
standard topology\footnote{This is a metrizable topology of uniform
  convergence on compact subsets of $\HH_h$.} of the space $H(\HH_h)$ of
analytic functions on $\HH_h$. We then characterize the set $\CS_h$ and obtain
stability of analytic continuation in the following sense\: if $\{f_n\}, f \subset \CS_h$ are such that $f_n \to f$ in $L^2(0,1)$, then $f_n \to f$ as $n \to \infty$ in $H(\bb{H}_{h})$. In Section~\ref{SECT minimizer} we study the properties of the minimizer of \eqref{lsq}.

Even though we have established well-posedness and stability of the
extrapolation problem, the above-mentioned results are not
quantitative, since they do not give rates of convergence of the extrapolation errors.
Figure~\ref{fig:illpos} (corresponding to a small value of the natural
regularization parameter) shows two perfectly admissible functions in $\PK_{h}$
that are virtually indistinguishable on $[0,1]$, but separate almost immediately beyond the
data window. 
\begin{figure}[t]
\center
\includegraphics[scale=0.33]{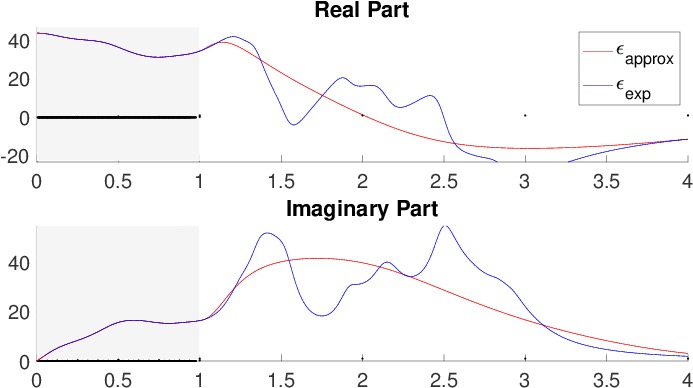}
\caption{Apparent ill-posedness of the extrapolation process.}
\label{fig:illpos}
\end{figure}
It suggests that the quantification of mathematical well-posedness is a matter of
practical importance. While there is no shortage of proposed algorithms for
extrapolation of experimental data in the vast literature on the subject,
there is no mathematically rigorous quantitative analysis of uncertainty
inherent in such extrapolation procedures.
We therefore consider two different functions $f$ and $g$ in $\PK_h$, that
differ by less than a small fraction $\epsilon$ of their size on the frequency
band $[0,1]$. Our goal is to estimate how much $f$ and $g$ can differ at a
given point $\Go_{0}>1$? We begin by giving a precise formulation of this question.
For any $\Ge>0$ we consider the set of  pairs
\[
U_{h}(\Ge)=\left\{(f,g) \in \PK_h:\frac{\|f-g\|_{L^2(0,1)}}{\max (\|\sigma_f\|, \|\sigma_g\|)} \leq \epsilon \right\},
\]
where $\sigma_f$ and $\sigma_g$ are the spectral measures in the
representation \eqref{intrep} of $f$ and $g$, respectively, and
$$\|\sigma_f\|:= \int_0^\infty \frac{d \sigma_f(\lambda)}{\lambda + 1}<+\infty$$
is finite interpreted as a "total norm" of $f$. Our goal is to find an upper
bound on the relative extrapolation error at the point $\Go_{0}$
\begin{equation} \label{Delta}
\Delta_{\Go_{0},h}(\epsilon) = \sup\left\{ \frac{|f(\Go_{0})-g(\Go_{0})|}{\max (\|\sigma_f\|, \|\sigma_g\|)}: (f,g) \in U_{h}(\Ge) \right\}.
\end{equation}
\begin{figure}[t]
  \centering
  \includegraphics[scale=0.2]{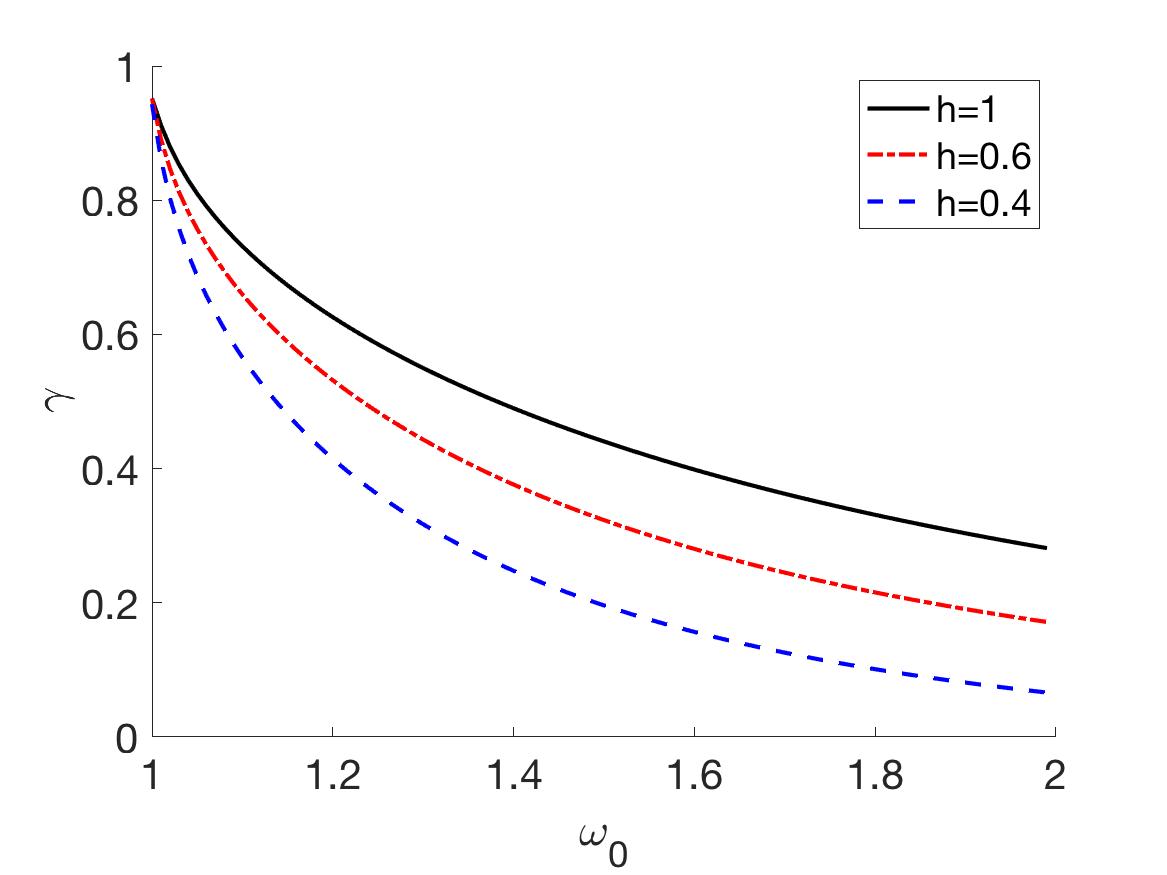}
  \caption{Power law exponent $\gamma$ as a function of $\Go$ for several
    values of $h$.}
  \label{fig:gammaofomega}
\end{figure}
The two fundamental questions determining the reliability of the extrapolation procedures are 
\begin{enumerate}
\item  Is it true that $\Delta_{\Go_{0},h}(\epsilon) \to 0$ as $\epsilon \to 0^+$?
\item What is the exact convergence rate of $\Delta_{\Go_{0},h}(\epsilon)$ to $0$?
\end{enumerate}
The first insight is the realization that, in fact, these questions are about
the difference $\phi=f-g$, rather than the pair $(f,g)$. The difference $\phi$
has the same spectral representation (\ref{Strep}), (\ref{intrep}) as $f$ and
$g$, except the spectral measure is no longer positive. Our next observation
is that the asymptotic behavior of $\GD_{\Go_{0},h}(\epsilon)$, as $\Ge\to 0$
is insensitive to certain restrictions on the spectral measures $\Gs$, as long
as the set of admissible measures is dense (in the weak-* topology) in the
space of measures (\ref{Strep}). For example, we may work only with absolutely
continuous measures with densities in $L^{2}(0,+\infty)$, permitting us to use
the theory of Hardy functions and Hilbert space methods to obtain exact
asymptotic behavior of $\GD_{\Go_{0},h}(\epsilon)$. The passage from pairs
$(f,g)$ to a single function $\phi=f-g$ is described in Section~\ref{SECT
  Reformulation}. The analysis of the Hilbert space problem for the difference
$\phi=f-g$ is in Section~\ref{SECT Optimal bound}, where it is shown that
$\GD_{\Go_{0},h}(\epsilon) \lesssim \epsilon^\gamma$ for some $\gamma \in
(0,1)$, giving a positive answer to our first question. The answer to the
second question is more nuanced, if we distinguish what we can prove
rigorously and what we can conjecture based on the numerical and analytical
evidence. The theory in Section~\ref{SECT Optimal bound} permits numerical
computation of the asymptotics of $\GD_{\Go_{0},h}(\epsilon)$ by relating it
to a similar problem without the symmetry constraint (property (a) from
Section~\ref{SECT prelim}).  Figure~\ref{fig:powerlaw} shows that
asymptotically $\GD_{\Go_{0},h}(\epsilon)\sim\Ge^{\gamma(\Go_{0}, h)}$, while
we also see from Figure~\ref{fig:comparison} that the symmetry requirement
does not change the value of the exponent $\gamma(\Go_{0}, h)$.
\begin{figure}[t]
\begin{center}
\begin{subfigure}[t]{2in}
\includegraphics[scale=0.22]{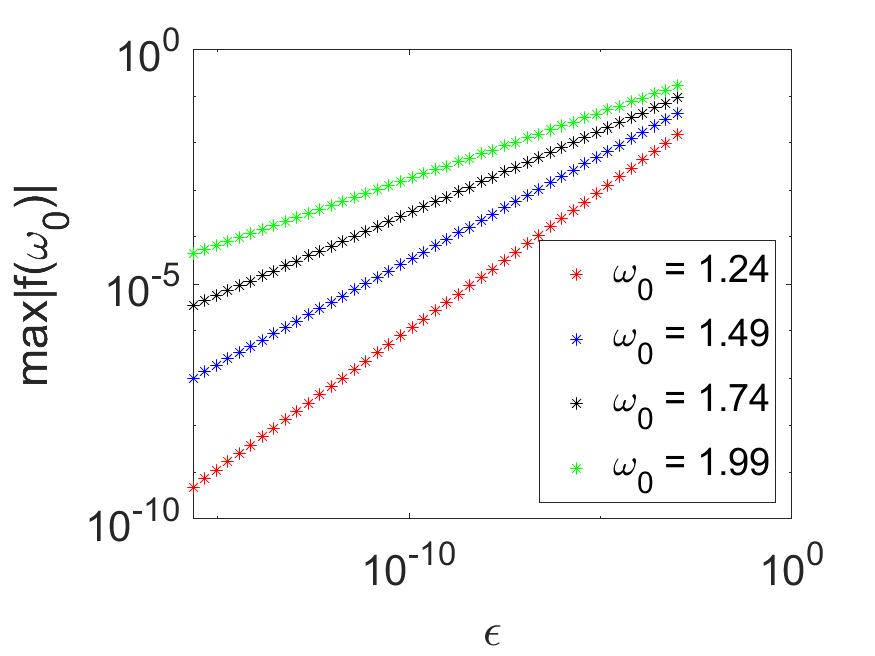}
\caption{}
\label{fig:powerlaw}
\end{subfigure}
\hspace{12ex}
\begin{subfigure}[t]{2in}
\includegraphics[scale=0.22]{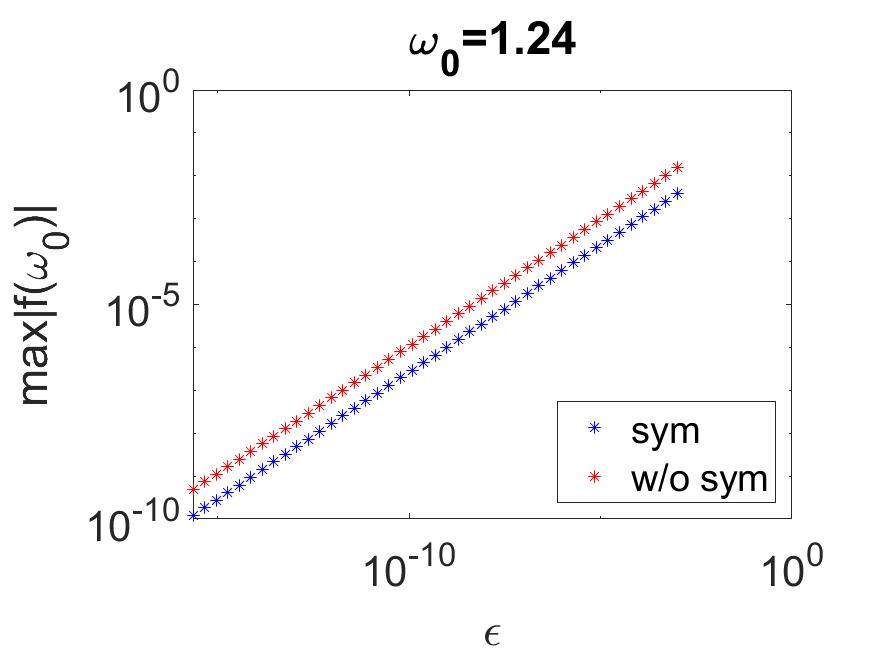}
\caption{}
\label{fig:comparison}
\end{subfigure}
\caption{Numerical support for the power law transition principle.}
\end{center}
\end{figure}

These results demonstrate the power law principle we
have formulated in \cite{grho-annulus, grho-gen}, generalizing the Nevanlinna
principle \cite{ciulli69,trefe19}. It says that the largest value a bounded analytic
function, which is of order $\Ge$ on a curve $\GG$ inside its domain of
analyticity can take at a point $\Go_{0}\not\in\GG$, decays as
$\Ge^{\Gg}$, where the exponent $0<\Gg<1$ depends on the geometry of the
domain, the curve $\GG$ and the point $\Go_{0}$. Figure~\ref{fig:gammaofomega}
shows how rapidly $\gamma(\Go_{0}, h)$ decays to 0, as
$\Go_{0}$ moves further away from $\GG$ for several values of $h$. The larger
the regularization parameter $h$ is, the better behaved is the extrapolation
problem.

In \cite{grho-gen,grho-annulus} we have gained some insight into the
mathematical structure of the maximizer function and the underlying mechanisms
that cause the power law precision deterioration in problems without the
symmetry constraint. Specifically, in the absence of symmetry the Hardy
function $\phi(z)$ of unit norm maximizing $|\phi(\Go_{0})|$ is a rescaled solution of a
linear integral equation of Fredholm type
\begin{equation}
\label{Kup}
\CK_{h} u + \epsilon^2 u = p_{\omega_0},
\end{equation}
where
\begin{equation} \label{K and p}
(\CK_{h} u)(\omega) = \int_{-1}^1 p_x(\omega) u(x) dx, 
\qquad \qquad p_{\Go_{0}}(\omega) = \frac{i}{2\pi (\omega-\overline{\Go_{0}}+2ih)}.
\end{equation}
The exponent $\gamma(\omega_0, h)$ can be computed from the unique solution
$u_{\epsilon}=u_{\epsilon, \omega_0, h}$ of the integral equation:
\begin{equation} \label{gamma main formula}
\gamma(\omega_0, h) = 1 - \lim_{\epsilon \to 0^+} \frac{\ln \|u_{\epsilon}\|_{L^2(-1,1)}}{\ln (1/\epsilon)}.
\end{equation}

The equality of the exponents for problems with and without symmetry shown in
Figure~\ref{fig:comparison} can be explained by the ``quantitative asymmetry''
of the solution $u_{\epsilon}$:
\begin{equation}
  \label{qassy}
  \lims_{\Ge\to 0}\frac{|u_{\epsilon}(\Go_{0})|}{|u_{\epsilon}(-\Go_{0})|}<1.
\end{equation}
Indeed, the symmetrized solution $v_{\epsilon}(\Go)=u_{\epsilon}(\Go)+\bra{u_{\epsilon}(-\bra{\Go})}$ has the
same order of magnitude at $\Go=\Go_{0}$ as $u_{\epsilon}(\Go_{0})$, as
$\Ge\to 0$. While numerically (\ref{qassy}) is seen to hold, we do not have a mathematical
proof of this inequality. Nonetheless, the equality of the exponents for
problems with and without symmetry is established in Section~\ref{SECT Optimal bound}.

\begin{figure}[t]
  \centering
  \includegraphics[scale=0.2]{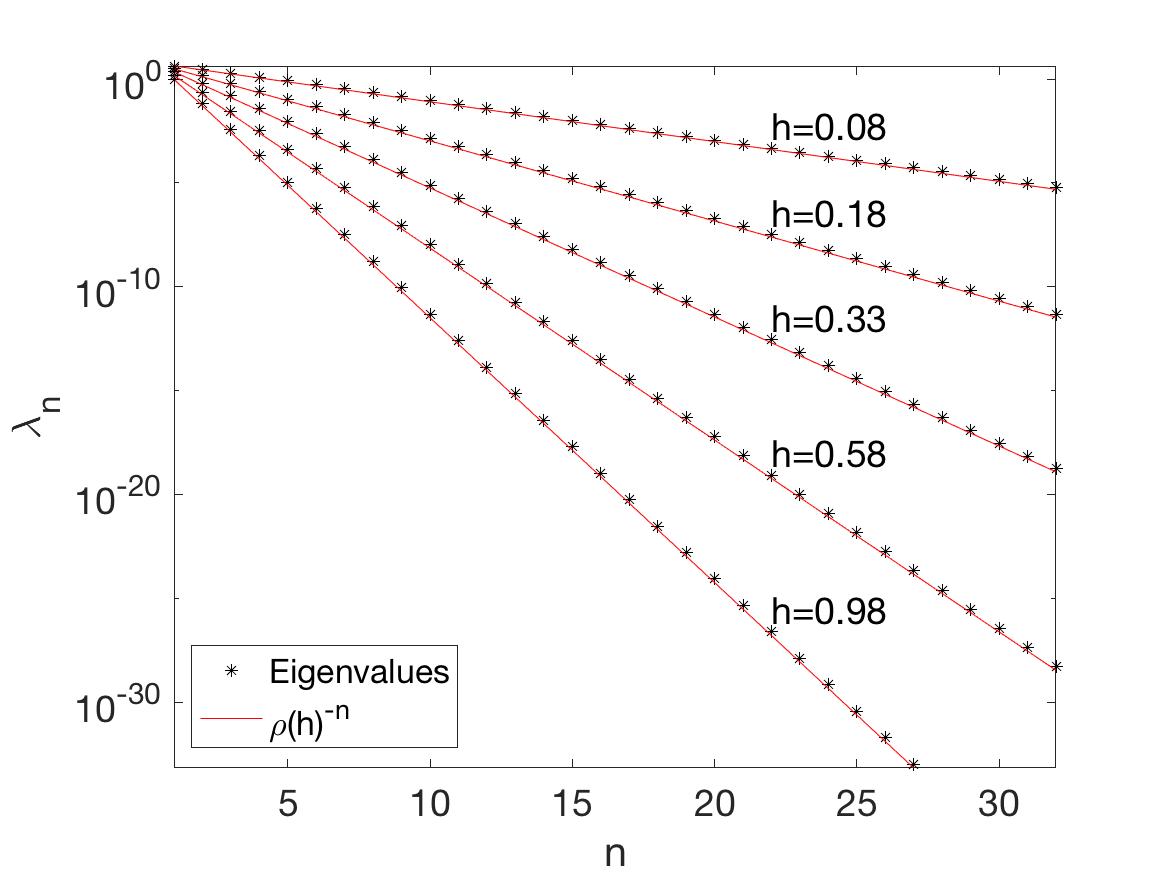}
  \caption{Comparison of the eigenvalues $\Gl_n$ of $\CK_h$ and $\rho(h)^{-n}$.}
  \label{fig:alpha}
\end{figure}

Once the symmetry constraint is discarded, the problem reduces to the one that
we have already studied in \cite{grho-gen}. The insights from that study
permit us to construct a ``near-optimal'' test function $\phi=f-g$ and give an
analytic formula for an upper bound on $\gamma(\omega_0, h)$, which is tight
for $h\ge 0.6$. To explain the construction of the near-optimal test function,
consider the orthonormal eigenbasis $\{e_{n}:n\ge 1\}\subset L^{2}(-1,1)$ of
$\CK_{h}$. We observe that taking $u=e_{n}$ in (\ref{K and p}) we obtain
\[
(p_{\Go_{0}},e_{n})_{L^{2}}=\bra{(\CK_{h}e_{n})(\Go_{0})}=\Gl_{n}\bra{e_{n}(\Go_{0})},
\]
where $\Gl_{n}>0$ are the corresponding eigenvalues. Then the solution of
(\ref{K and p}) can be written as
\[
u_{\epsilon}(\Go)=\sum_{n=1}^{\infty}\frac{\Gl_{n}\bra{e_{n}(\Go_{0})}e_{n}(\Go)}{\Gl_{n}+\Ge^{2}}.
\]
The next idea comes from the upper bound on the decay of the eigenvalues
$\Gl_{n}$ from \cite{beto17} and an identical asymptotics from
\cite{parfenov}. Figure~\ref{fig:alpha} shows that $\Gl_{n}\sim\rho^{-n}$,
where $\rho$ is the Riemann invariant of
$G_{h}=\bb{C}_{\infty}\setminus([-1,1]\pm ih)$. The Riemann invariant of a
doubly-connected region is the unique value of $\rho>1$ such that $G_{h}$ is
conformally equivalent to the annulus
\[
A_{\rho}=\{z\in\bb{C}:\rho^{-1/2}<|z|<\rho^{1/2}\}.
\]
If $\Psi:G_{h}\to A_{\rho}$ is the conformal isomorphism, then it maps
$\GG_{h}=[-1,1]+ih$ onto the circle $|z|=\rho^{-1/2}$ and the real
line\footnote{In order to explain the structure of the maximizer function it
  is convenient to work in a shifted plane $\bb{H}_{h}+ih$, so that the
  interval $[-1,1]$ where frequencies are measured corresponds to $\GG_{h}$
  and the boundary of analyticity $\im\Go=-h$ shifts to the real line.} is
mapped to the unit circle. In the annulus $A_{\rho}$ the same question we are
studying in the upper half-plane can be analyzed completely (see
\cite{grho-annulus} for details). In $A_{\rho}$ the eigenfunctions of the
corresponding integral operator are just functions $z^{n}$. Even though it is not
true that the eigenfunctions of $\CK_{h}$ are $\Psi(\Go)^{n}$, we can treat
them as such, replacing $e_{n}(\Go)$ with
$\Tld{e}_{n}(\Go)=(\sqrt{\rho}\Psi(\Go))^{n}$ (so that $|\Tld{e}_{n}(\Go)|=1$
on $\GG_{h}$). This gives us a replacement
\begin{equation}
  \label{solstruct}
  \Tld{u}_{\epsilon}(\Go)=\sum_{n=1}^{\infty}\frac{\bra{\Psi(\Go_{0})}^{n}\Psi(\Go)^{n}}{\rho^{-n}+\Ge^{2}}
\end{equation}
for the solution $u_{\epsilon}(\Go)$ of (\ref{K and p}). Lemma~\ref{LEM series asymp}
below shows that
\[
\Tld{u}_{\epsilon}(\Go_{0})=\sum_{n=1}^{\infty}\frac{|\Psi(\Go_{0})|^{2n}}{\rho^{-n}+\Ge^{2}}\sim
\Ge^{-2\Gth_{0}}P\left(\frac{2\ln(1/\Ge)}{\ln\rho}\right),
\]
where
\[
P(t)=\left(\frac{\rho}{|\Psi(\Go_{0})|^{2}}\right)^{t}\sum_{k\in\bb{Z}}\frac{|\Psi(\Go_{0})|^{2k}}{\rho^{t}+\rho^{k}}
\]
is a smooth 1-periodic function of $t$, and
\[
\Gth_{0}=1+\frac{2\ln|\Psi(\Go_{0})|}{\ln\rho}.
\]
\begin{figure}[t] 
\begin{center}
\begin{subfigure}[t]{2in}
\includegraphics[scale=0.15]{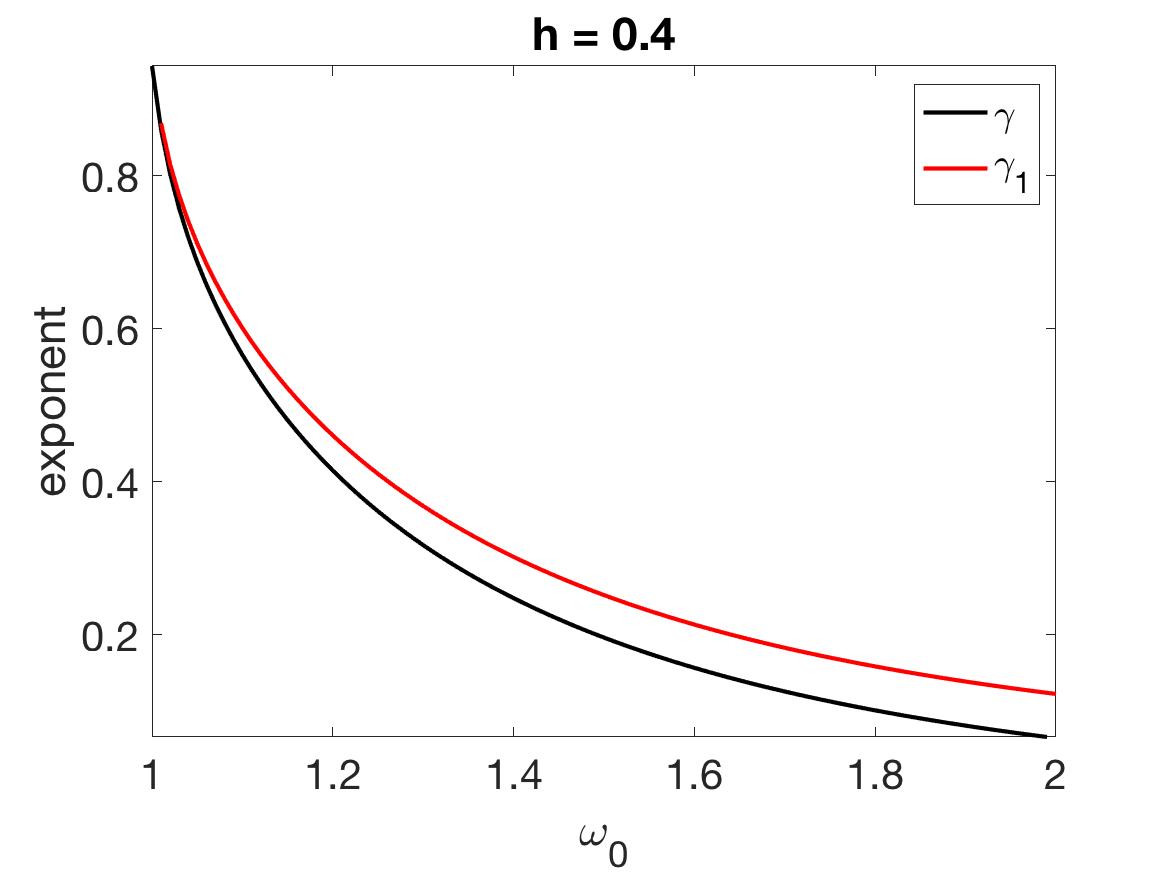}
\caption{}
\end{subfigure}
\hspace{.15in}
\begin{subfigure}[t]{2in}
\includegraphics[scale=0.15]{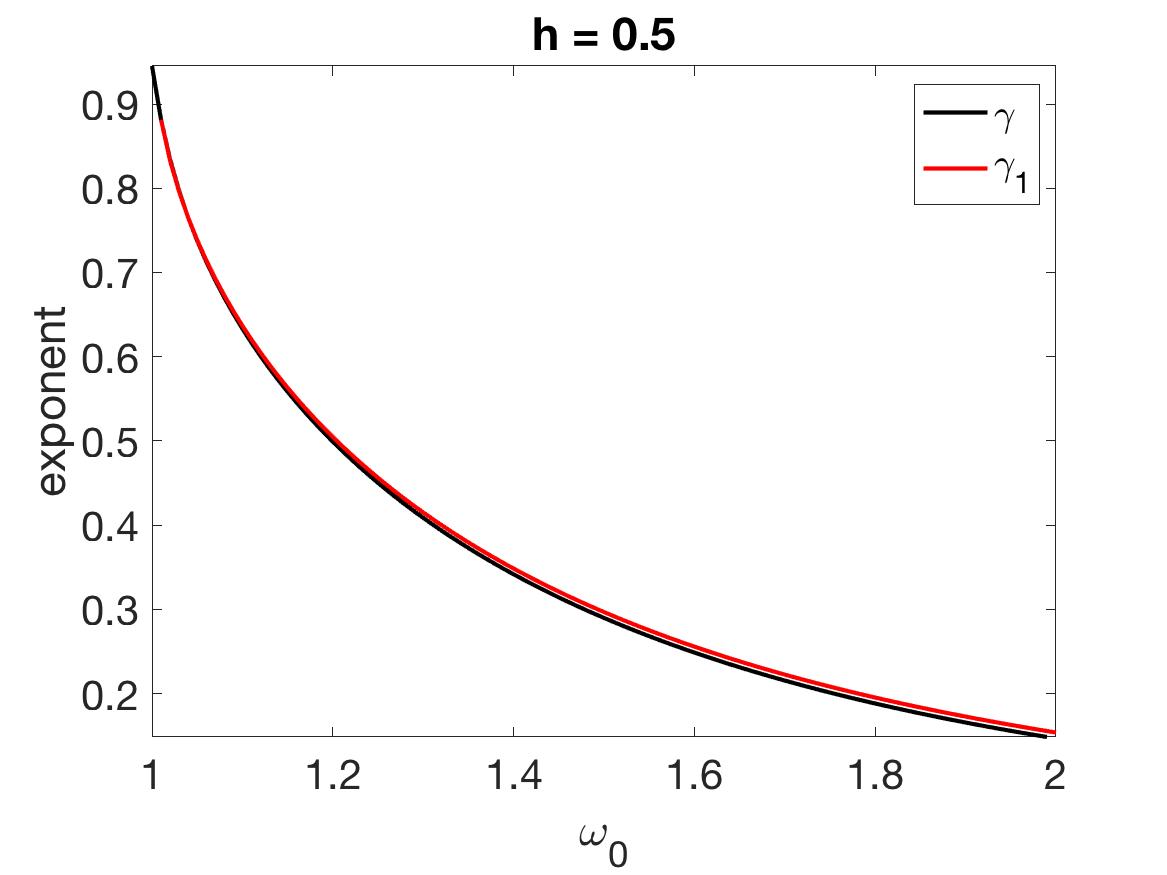}
\caption{}
\end{subfigure}
\hspace{.15in}
\begin{subfigure}[t]{2in}
\includegraphics[scale=0.15]{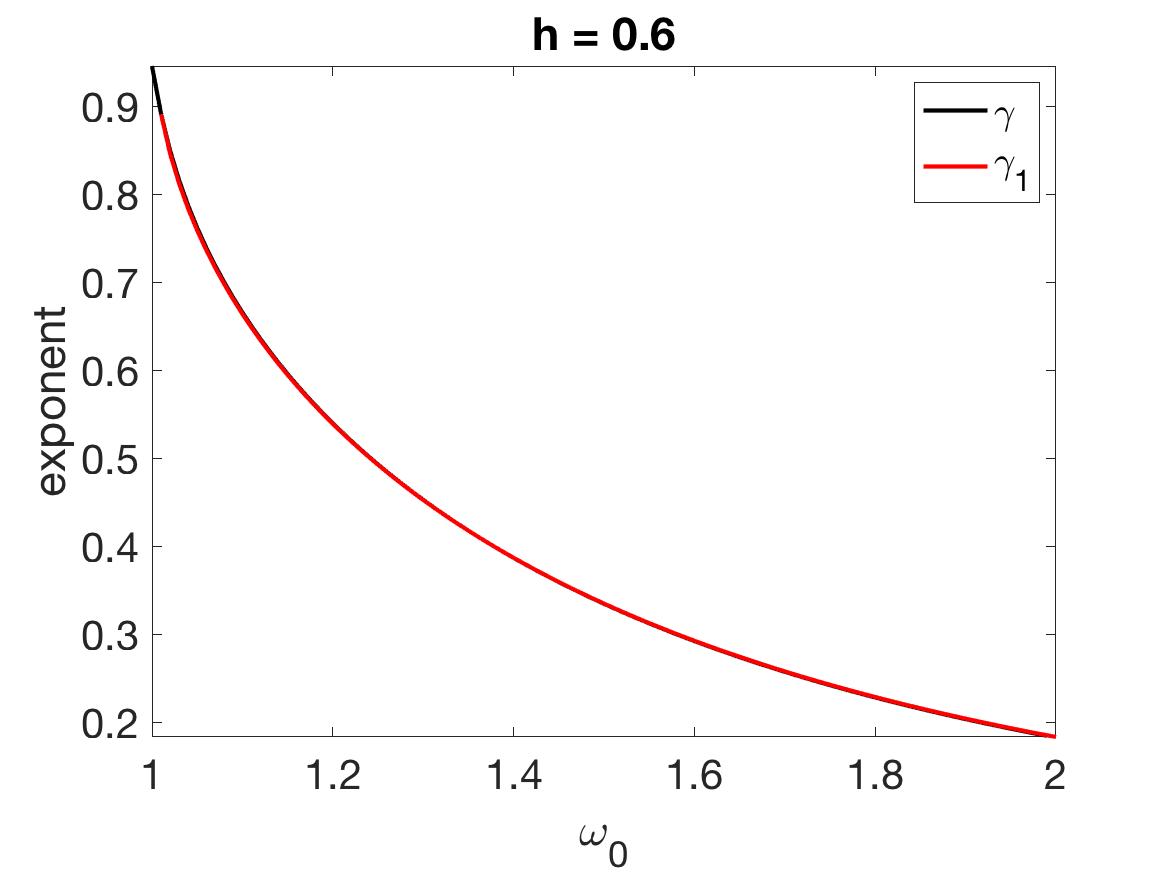}
\caption{}
\end{subfigure}
\caption{Comparison of $\gamma$ and $\gamma_1$}
\label{fig:gamma}
\end{center}
\end{figure}

The same lemma shows that when $\Go\in\GG_{h}$, then
$|\Psi(\Go)|=\rho^{-1/2}$, and
\[
|\Tld{u}(\Go)|\sim\Ge^{-2\Gth_{h}},\qquad\Gth_{h}=\hf+\frac{\ln|\Psi(\Go_{0})|}{\ln\rho},
\]
while, when $\Go\in\bb{R}$, $|\Psi(\Go)|=1$, and we have
\[
|\Tld{u}(\Go)|\sim\Ge^{-2\Gth_{\bb{R}}},\qquad\Gth_{\bb{R}}=1+\frac{\ln|\Psi(\Go_{0})|}{\ln\rho},
\]
Then 
$
M(\Go)=\Ge^{2\Gth_{\bb{R}}}\Tld{u}(\Go)
$
is $O(1)$ on $\bb{R}$,  $O(\Ge)$ on
$\GG_{h}$ and $O(\Ge^{\Gg_{1}})$ at $\Go_{0}$, where
\[
\Gg_{1}(\Go_{0})=2(\Gth_{\bb{R}}-\Gth_{0})=-\frac{\ln|\Psi(\Go_{0})|}{\ln\rho}.
\]
The explicit formula for the conformal isomorphism $\Psi:G_{h}\to A_{\rho}$
has been derived in \cite[p.~138]{akhi90} in terms of elliptic functions and
integrals, permitting us to compute an
upper bound $\Gg_{1}(\Go_{0})$ on the true exponent $\Gg(\Go_{0})$.  Figure~\ref{fig:gamma}
shows that $\Gg_{1}(\Go_{0})$ is a very good approximation for $\gamma$, when $h\ge0.6$.

\begin{lemma} \label{LEM series asymp}
Let $a\in\bb{C}$ and $b>0$ be such that $0<b<|a|<1$. Let
\begin{equation}
  \label{maxstruct}
  \phi(\eta)=\sum_{n=0}^{\infty}\frac{a^{n}}{\eta+b^{n}}.
\end{equation}
Then the asymptotics of $\phi(\eta)$, as $\eta\to 0^{+}$ is surprisingly
irregular, depending on the limit
\[
t=\lim_{j\to\infty}\left\{\frac{\ln\eta_j}{\ln b}\right\}
\]
along a sequence $\eta_{j}\to 0$, as $j\to\infty$, where $\{x\}$ denotes the
fractional part of $x$. Specifically,
$$\phi(\eta_j) \sim \phi_0(t) \eta_j^{-\gamma},$$
where
$$\phi_0(t) = \frac{b^{t}}{a^{t}}
\sum_{k \in \ZZ} \frac{a^{k}}{b^{t}+b^{k}}$$ 
is a smooth 1-periodic function, and
\[
\Gg=1-\frac{\ln a}{\ln b}.
\]
In the formulas above $a^{t}=e^{t\ln a}$ and $\ln$ can denote
any analytic branch (independent of $\eta$) that agrees with usual
logarithm for positive real numbers.
\end{lemma}

\begin{proof}
We first notice that unlike $\phi(\eta)$, the function 
\[
\psi(\eta)=\sum_{n=1}^{\infty}\frac{a^{-n}}{\eta+b^{-n}}
\]
is regular at $\eta=0$. In fact, $\psi(0)=b/(a-b)$. We therefore define a new function

\[
F(\eta)=\sum_{n \in \ZZ} \frac{a^{n}}{\eta+b^{n}}=\phi(\eta)+\psi(\eta),
\]
which obviously satisfies
\[
\lim_{j\to\infty}F(\eta_{j})\eta_{j}^{\Gg}=\lim_{j\to\infty}\phi(\eta_{j})\eta_{j}^{\Gg},
\]
whenever $\eta_{j}\to 0^{+}$ and the limit on the \rhs\ exists. Introducing the integer and fractional parts 

\[
N(\eta)=\left[\frac{\ln\eta}{\ln b}\right],\qquad
\Ga(\eta)=\left\{\frac{\ln\eta}{\ln b}\right\}
\]

\noindent we make a change of index of summation $k=n-N(\eta)$ and obtain, using
\[
N(\eta)=\frac{\ln\eta}{\ln b}-\Ga(\eta),
\]
after a short calculation, that
\[
F(\eta)\eta^{\Gg}=\sum_{k \in \ZZ} \frac{a^{k-\Ga(\eta)}}{1+b^{k-\Ga(\eta)}}=
\frac{b^{\Ga(\eta)}}{a^{\Ga(\eta)}}\sum_{k \in \ZZ} \frac{a^{k}}{b^{\Ga(\eta)}+b^{k}}.
\]
The statement of the lemma is now apparent.
\end{proof}

In general, we have shown in \cite{grho-annulus,grho-gen} that the exact
exponent $\Gg(\Go_{0},h)$ is determined by the exponential decay of the
magnitudes $|e_n(\Go_{0})|$ of the orthonormal eigenbasis $e_{n}$ of the
integral operator $\CK_{h}$. Specifically, we have proved that if
\begin{equation} \label{exp decay}
\lambda_n \simeq e^{-\alpha(h) n}, \qquad
|e_n(\Go_{0})| \simeq e^{-\beta(\Go_{0},h) n},
\end{equation}
then $0<2 \beta(\Go_{0},h) < \alpha(h)$, and
\begin{equation}
  \label{gamma}
\Gg(\Go_{0},h)=\frac{2 \Gb(\Go_{0},h)}{\Ga(h)}.
\end{equation}

The conjectured asymptotics $\Gl_{n}\sim\rho^{-n}$ of (squares of) singular
values of the restriction operator $\CR_{h}$ exactly coincides with the
asymptotics of the restriction operators to smooth domains established in
\cite{parfenov}. Unfortunately, the methods in \cite{parfenov} are not
applicable, since the end-points of the interval $[-1,1]$ can be regarded as
corners of angle 0, violating the desired smoothness
requirements. Nonetheless, Figure~\ref{fig:alpha} indicates that the technical
assumptions in \cite{parfenov} on the smoothness of domains could probably be
significantly relaxed.

The eigenvalues $\lambda_n$ are also connected to Kolmogorov $n$-widths
\cite{pinkus85}, since they are squares of singular values of the restriction
operator $\CR_{h} : H^2(\HH_h) \to L^2(-1,1)$ (here $H^2$ is defined in
\eqref{H^2 space def}). Specifically (cf. \cite[Theorem~6.1]{fisher83}),
$\sqrt{\lambda_{n+1}}$ is the Kolmogorov $n$-width of the restriction to
$L^2(-1,1)$ of closed unit ball in $H^2(\HH_h)$. The relation of the
Kolmogorov $n$-widths of restrictions of various classes of analytic functions
to corresponding Riemann invariants have been know in many cases
\cite{erokhin68, zakski76,fimi80}.

\section{The least squares problem} 
\setcounter{equation}{0} 
\label{SECT lsq}

\subsection{Existence and uniqueness}
We begin by examining the existence and uniqueness questions in the least
squares problem (\ref{lsq}). Let $f_{n}\in\PK_{h}$ be a minimizing sequence in (\ref{lsq}). Then
it has to be bounded in the $L^{2}(0,1)$ norm. We will show that this implies
existence of a subsequence converging uniformly on compact
subsets of $\HH_{h}$ to an analytic function. In general, this limit
does not need to be in $\PK_{h}$, since it is not closed in $H(\HH_{h})$. We will, therefore, need to characterize
the closure $\bra{\PK_{h}}$ of $\PK_{h}$.

We recall that a family of functions in $H(G)$ is called normal, if every
sequence has a convergent in $H(G)$ subsequence. In other words, normal families
of functions are exactly the precompact subsets in $H(G)$.
\begin{theorem}
  \label{th:closure}~
  \begin{enumerate}
\item[(i)] The closure of $\PK_{h}$ in $H(\bb{H}_{h})$ is
$
\CS_{h}=\{f(\Go)=F((\Go+ih)^{2}):F\in\mathfrak{S}\}.
$
\item[(ii)] For any $M>0$ the family of functions 
$
\CS_{h}^{M}=\{f\in\CS_{h}:\|f\|_{L^{2}(0,1)}\le M\}
$
is normal.
  \end{enumerate}
\end{theorem}
\begin{proof}
The proof is based on the representation (\ref{Strep}), where we interpret the
measure $\Gs$ as an element of the Banach space $\CB^{*}$ dual to
\[
\CB=\left\{\phi\in C([0,+\infty)): \lim_{\Gl\to\infty} \Gl\phi(\Gl)=0 \right\},
\]
with the norm
\[
\|\phi\|_{\CB}=\max_{\Gl\ge 0}(\Gl+1)|\phi(\Gl)|.
\]
If we define the action of the measure $\Gs$ on $\phi\in\CB$ by
\[
\av{\phi,\Gs}=\int_{0}^{\infty}\phi(\Gl)d\Gs(\Gl),
\]
then
\begin{equation} \label{sigma norm}
\|\Gs\|_{*}=\int_0^\infty \frac{d\Gs(\Gl)}{\Gl+1},
\end{equation}
when the measure $\Gs$ is nonnegative.

The conclusion of the theorem then follows easily from the fundamental
estimate in the lemma below.
  \begin{lemma} \label{LEM technical}
There exists $c_{h}>0$ and $C_h>0$ depending only on $h$, such that for every $f \in \CS_h$
\begin{equation*}
c_h \|f\|_{L^2(0,1)} \le \rho + \|\Gs\|_{*} \leq C_h \|f\|_{L^2(0,1)},
\end{equation*}
where
\[
\rho=\lim_{\Go\to\infty}f(\Go).
\]
\end{lemma}

\begin{proof}
Let us start by proving the second inequality. Applying the
H{\"o}lder inequality to the representation
\begin{equation}
  \label{Lhrep}
  f(\Go)=\rho+\int_{0}^{\infty}\frac{d\Gs(\Gl)}{\Gl-(\Go+ih)^{2}}
\end{equation}
we obtain
\begin{equation*}
\|f\|_{L^2(0,1)} \geq \left( \int_0^1 | \re(f) |^2 d \omega \right)^\frac{1}{2} \geq \left| \int_0^1 \re(f) d \omega \right|.
\end{equation*} 

\noindent Applying Fubini's theorem we then compute

\begin{equation*}
\int_0^1 \re(f) d \omega = \rho +  \int_0^1 \int_0^\infty \re \left( \frac{1}{\lambda - (\omega+ih)^{2}} \right) d \sigma(\lambda) d \omega  = \rho + \int_0^\infty \varphi(\sqrt{\lambda}) \frac{d \sigma(\lambda)}{\lambda + 1}, 
\end{equation*}

\noindent where

$$\varphi(x)= \frac{x^2+1}{4x} \ln{\left( 1+\frac{4x}{(x-1)^2+h^2} \right)}.$$

\noindent Note that $\varphi(x)>0$ for $x>0$, and because $\ln(1+x) \sim x$ as $x\rightarrow 0$ we get

\begin{equation*}
\begin{split}
\lim_{x\rightarrow 0} \varphi(x) =  \frac{1}{1 + h^2} > 0,
\qquad \qquad
\lim_{x\rightarrow \infty} \varphi(x) = 1 >0.
\end{split}
\end{equation*}

\noindent Thus $\inf_{[0,\infty)} \varphi(x) = \mu_{h} >0$, which implies the
desired estimate with $C_{h}=1/\mu_{h}$.

Let us now turn to the first inequality. Again, by H{\"o}lder's inequality

\begin{equation*}
\begin{split}
\frac{1}{2}\|f\|_{L^2(0,1)}^2 - \rho^2 \leq&  \int_0^1 \left( \int_0^\infty  \frac{d \sigma(\lambda)}{\left| \lambda - (\omega+ih)^{2} \right|}  \right)^2 d \omega 
\leq
\\
\leq&
\int_0^\infty \frac{d \sigma(\lambda)}{\lambda+1} \cdot \int_0^1 \int_0^\infty \frac{\lambda+1}{\left| \lambda - (\omega+ih)^{2} \right|^2}  d \sigma(\lambda) d \omega =\|\sigma\|_* \cdot \int_0^\infty \psi(\lambda) d \sigma(\lambda)
\end{split}
\end{equation*}

\noindent where

\begin{equation*}
\psi(\lambda) = \int_0^1  \frac{\lambda + 1 }{\left| \lambda - (\omega+ih)^{2} \right|^2} d \omega = \frac{\varphi(\sqrt{\lambda})}{\lambda + h^2} + \frac{\lambda + 1}{4h(\lambda + h^2)} \left( \arctan{\tfrac{\sqrt{\lambda}+1}{h}} -\arctan{\tfrac{\sqrt{\lambda}-1}{h}} \right).
\end{equation*}

\noindent Note that $(\lambda+1) \psi(\lambda)$ is bounded in $[0,\infty)$, because $\varphi$ is a bounded function and the difference of arctangents can be bounded by $\frac{2h}{\lambda-1}$ for $\lambda>1$, by the mean value theorem. But then the desired inequality follows from the estimate

\begin{equation*}
\int_0^\infty \psi(\lambda) d \sigma(\lambda) \leq C_h \int_0^\infty \frac{d \sigma(\lambda)}{\lambda+1} = C_h \|\sigma\|_*. 
\end{equation*}

\end{proof}

Obviously $\PK_h \subset \CS_h$ and Theorem~\ref{th:closure} follows from the next lemma.

\begin{lemma} \label{LEM L_h is closed}
\mbox{}
\begin{enumerate}
\item[(i)] $\displaystyle \CS_h$ is closed in $H(\HH_h)$.

\item[(ii)] $\CS_h \subset \bra{\PK_h}$

\end{enumerate}
\end{lemma}

\begin{proof}
(i) Let $\{f_n\} \subset \CS_h$ be a sequence such that $f_n \to f$ in
$H(\HH_h)$. Then according to Lemma~\ref{LEM technical} the sequences
$\{\rho_n\} \subset \RR$ and $\{\sigma_n\} \subset \CB^*$ are bounded. By the
Banach-Alaoglu theorem the closed unit ball in $\CB^*$ is compact in the
weak-* topology. It is also sequentially compact because the Banach space $\CB$ is separable.   
Thus, there exist subsequences (which we do not relabel) $\rho_n \to \rho$ and  $\sigma_n \stackrel{\ast}{\rightharpoonup} \sigma$ weakly-* in $\CB^*$. Let us write

\begin{equation*}
f_n(\omega) = \rho_n + \|\sigma_n\|_* + \int_0^\infty G(\omega, \lambda) d \sigma_n(\lambda) ,
\end{equation*}

\noindent where

\begin{equation*}
G(\omega, \lambda) = \frac{1}{\lambda - (\omega+ih)^2} - \frac{1}{\lambda+1} = \frac{1+(\omega+ih)^2}{(\lambda-(\omega+ih)^2) \ (\lambda+1)}.
\end{equation*}

\noindent It is now evident that $G(\omega, \cdot) \in \CB$ for each fixed $\omega \in \HH_h$. Upon extracting convergent subsequence of the bounded sequence $\{\|\sigma_n\|_*\}$, with limit denoted by $a$, we obtain that

\begin{equation*}
\begin{split}
f(\omega) = \lim_{n \rightarrow \infty} f_n(\omega) = \rho + a + \int_0^\infty G(\omega,\lambda) d \sigma(\lambda)
= \rho + a - \|\sigma\|_*  + \int_0^\infty \frac{d \sigma(\lambda)}{\lambda - (\omega+ih)^2} .
\end{split}
\end{equation*}

\noindent By lower semicontinuity of the norm $a \geq \|\sigma\|_*$, hence we conclude that $f \in \CS_h$.

(ii)
1. Let us start by showing that for any constant $\rho \geq 0$, there exists $\{g_n\} \subset \PK_h$ such that $g_n \to \rho$ uniformly on $[0,1]$ as $n \to \infty$. Indeed, define

\begin{equation*}
g_n(\omega) = \rho \int_{n}^{n+1} \frac{\lambda d \lambda}{\lambda-(\omega+ih)^2}.
\end{equation*}

\noindent Clearly, $g_n \in \PK_h$ and

\begin{equation*}
g_n(\omega) - \rho = \rho (\omega+i h)^2 \int_n^{n+1} \frac{d \lambda}{\lambda-(\omega+ih)^2},
\end{equation*}

\noindent which approaches to zero, as $n \to \infty$, uniformly on compact subsets of $\HH_h$.

2. Let now $f \in \CS_h$ and let $\rho$ and $\sigma$ be as in its definition. Consider the functions

\begin{equation*}
h_n(\omega) = \int_0^n \frac{d \sigma(\lambda)}{\lambda - (\omega+ih)^2}.
\end{equation*} 

\noindent Note that $h_n \in \PK_h$, since its corresponding measure is $d\sigma_n = \chi_{(0,n)} d\sigma$ and

\begin{equation*}
\int_0^\infty d \sigma_n(\lambda) = \int_0^n d \sigma(\lambda) \leq (n+1) \int_0^n \frac{d \sigma(\lambda)}{\lambda + 1} < \infty.
\end{equation*} 

\noindent Now

\begin{equation*}
f(\omega) - h_n(\omega) = \rho + \int_n^\infty \frac{d \sigma(\lambda)}{\lambda - (\omega+ih)^2}
\end{equation*}

\noindent and by dominated convergence the above difference tends to $\rho$ uniformly on compact subsets of $\HH_h$. It remains to use the sequence $\{g_n\}$ from part 1 to get that $g_n + h_n$ is the desired sequence in $\PK_h$ converging to $f$ in $H(\HH_h)$. 

\end{proof}
To prove part (ii) of Theorem~\ref{th:closure} we observe that for any compact
subset $K\subset\bb{H}_{h}$ there exists a constant $C_{K}$ so that
\[
C_{K}=\sup_{\Gl\ge 0}\sup_{\Go\in K}\frac{\Gl+1}{|\Gl-(\Go+ih)^{2}|}<+\infty.
\]
Thus, for any $\Go\in K$  and $f\in\CL_{h}$ we have from representation (\ref{Lhrep})
\[
|f(\Go)|\le\rho+C_{K}\|\Gs\|_{*}.
\]
Now, Lemma~\ref{LEM technical} implies that the family of functions
$\CL_{h}^{M}$ is locally equibounded. We conclude, by Montel's theorem, that
$\CL_{h}^{M}$ is a normal family of analytic functions.
\end{proof}

A corollary of Theorem~\ref{th:closure} is stability of analytic continuation.
\begin{corollary} \label{COR stability}
Let $\{f_n\}, f \subset \CS_h$ be such that $f_n \to f$ in $L^2(0,1)$, then
$f_n \to f$ as $n \to \infty$ in $H(\bb{H}_{h})$.
\end{corollary}
Indeed, if $f_n \to f$ in $L^2(0,1)$, then $\|f_{n}\|_{L^{2}(0,1)}$ is
bounded. Then any converging subsequence $f_{n_{k}}\to g$ in $H(\bb{H}_{h})$ must
also converge to $g$ in $L^2(0,1)$. But then $f=g$ on $(0,1)$. Since both $f$
and $g$ are analytic in $\bb{H}_{h}$, then $f=g$ everywhere. Since the set of
limits of converging subsequences of $f_{n}$ consists of a single element
$\{f\}$, we conclude that $f_{n}\to f$ in $H(\bb{H}_{h})$.

Let us now return to the least squares problem (\ref{lsq}).
\begin{theorem} \label{THM lsq}
For a given $f_{\rm exp} \in L^2(0,1)$, the least squares problem
\begin{equation} \label{least squares 2}
\mathfrak{E} = \mathfrak{E}(f_{\rm exp}) = \min_{f\in \CS_h} \|f-f_{\rm exp}\|_{L^2(0,1)}
\end{equation}
 has a unique solution. Moreover,
\[
\inf_{f\in \PK_h} \|f-f_{\rm exp}\|_{L^2(0,1)}=\mathfrak{E}(f_{\rm exp}).
\]
\end{theorem}

\begin{proof}
To prove existence, let $\{f_{n}\}_{n=1}^\infty \in \CS_h$ be a minimizing sequence, then it is bounded in $L^{2}(0,1)$. Let us extract a weakly convergent subsequence, not relabeled, $f_{n} \weak f_{0}$ in $L^{2}(0,1)$, as $n\to\infty$. The limiting function $f_{0}$ is in $\CS_h$. By the
convexity of the $L^{2}$-norm we have
\[
\mathfrak{E} = \limi_{n\to\infty}\|f_{n}-f_{\rm exp}\|_{L^2(0,1)} \ge \|f_{0}-f_{\rm exp}\|_{L^2(0,1)}.
\]
Hence, $f_0$ is a minimizer. To prove that the infimum in \eqref{least squares 2} stays the same if we replace $\CS_h$ by $\PK_h$ we note that if $f_0 \in \CS_h$ is a minimizer, then there exists a sequence $\{g_n\} \subset \PK_h$ converging to $f_0$ strongly in $L^2(0,1)$.

To prove uniqueness, let $f_{1}$ and $f_{2}$ be two different solutions. Then $\|f_j-f_{\rm exp}\|_{L^2(0,1)}=\mathfrak{E}$ for $j=1,2$. Observe that the function $f_{t}=tf_{1}+(1-t)f_{2}$ is also admissible and therefore
\[
\mathfrak{E} \le \|f_{t}-f_{\rm exp}\|_{L^2(0,1)} \le 
t\|f_{1}-f_{\rm exp}\|_{L^2(0,1)} + (1-t)\|f_{2}-f_{\rm exp}\|_{L^2(0,1)}=\mathfrak{E},
\]
therefore $\|f_{t}-f_{\rm exp}\|_{L^2(0,1)} = \mathfrak{E}$ for
all $t\in[0,1]$. However,
\[
\|f_{t}-f_{\rm exp}\|^{2}_{L^2(0,1)} = t^{2}\|f_{1}-f_{2}\|_{L^2(0,1)}^{2} + 2 t \Re (f_{1}-f_{2},f_{2}-f_{\rm exp})
+\|f_{2}-f_{\rm exp}\|_{L^2(0,1)}^{2},
\]
which cannot be constant, since the coefficient at $t^{2}$ is
non-zero by our assumption $f_{1}\not=f_{2}$. The obtained contradiction,
concludes the theorem.
\end{proof}

\subsection{Properties of the minimizer} \label{SECT minimizer}
In this section we will prove that if the minimum in (\ref{least squares 2}) is nonzero,
then the minimizer must be a rational function in $\bb{C}$ with poles (and
zeros) on the line $\im(\Go)=h$. We use the method of Caprini \cite{capr74,capr80} to
prove the statement. The method for finding the necessary and sufficient
conditions for a minimizer in (\ref{least squares 2}) is based on our ability
to compute the effect of the change of $\rho$ and spectral measure $\Gs$ in
representation (\ref{Strep}) on the value of the functional we want to
minimize.  Suppose that
\[
f_{*}(\Go)=\rho_{*}+\int_{0}^{\infty}\frac{d\Gs_{*}(\Gl)}{\Gl-(\Go+ih)^{2}}
\]
is the minimizer and
\begin{equation}
  \label{compet}
  f(\Go)=\rho+\int_{0}^{\infty}\frac{d\Gs(\Gl)}{\Gl-(\Go+ih)^{2}}
\end{equation}
is a competitor. The variation $\phi=f-f_{*}$ can then be written as
\[
\phi(\Go)=\GD\rho+\int_{0}^{\infty}\frac{d\nu(\Gl)}{\Gl-(\Go+ih)^{2}},\qquad
\nu=\Gs-\Gs_{*},\quad\GD\rho=\rho-\rho_{*}.
\]
We then compute
\begin{equation}
  \label{var}
    \|f-f_{\rm exp}\|_{L^{2}}^{2}-\|f_{*}-f_{\rm exp}\|_{L^{2}}^{2}=
\GD\rho\lim_{t\to\infty}tC(t)+\int_{0}^{\infty}C(t)d\nu(t)+\|\phi\|_{L^{2}}^{2},
\end{equation}
where
\begin{equation}
  \label{Caprini}
C(t)=2\re\int_{0}^{1}\frac{f_{*}(\Go)-f_{\rm exp}(\Go)}{t-(\Go-ih)^{2}}d\Go,\qquad  t\ge 0
\end{equation}
is the Caprini function of $f_{*}(\Go)$.

\begin{theorem}
  \label{th:Caprini}
Suppose the infimum in (\ref{lsq}) is nonzero, then the
 the minimizer $f_{*}\in\CS_{h}$ in (\ref{least squares 2}) is given by
 \begin{equation}
   \label{fstar}
   f_{*}(\Go)=\rho_{*}+\sum_{j=1}^{N}\frac{\Gs_{j}}{t_{j}-(\Go+ih)^{2}}
 \end{equation}
 for some for some $N\ge 0$, $\Gs_{j}>0$, $0\le t_{1}<t_{2}<\dots<t_{N}$ and
 $\rho_{*}\ge 0$.  Moreover, $f_{*}$, given by (\ref{fstar}) is the minimizer
 \IFF its Caprini function $C(t)$ is nonnegative and vanishes at $t=t_{j}$,
 $j=1,\ldots,N$, and ``at infinity'', in the sense that
\begin{equation}
  \label{zatinf}
  2\re\int_{0}^{1}(f_{\rm exp}(\Go)-f_{*}(\Go))d\Go=\lim_{t\to\infty}tC(t)=0,
\end{equation}
provided $\rho_{*}>0$.
\end{theorem}
\begin{proof}
If $\rho_{*}>0$, then we can consider the competitor (\ref{compet}) with
$\Gs=\Gs_{*}$. Formula (\ref{var}) then implies that
\[
\GD\rho\lim_{t\to\infty}tC(t)+(\GD\rho)^{2}\ge 0,
\]
where $\GD\rho$ can be either positive or negative and can be chosen as small
in absolute value as we want. This implies (\ref{zatinf}).

  Next, suppose $t_{0}\in[0,+\infty)$ is in the support of $\Gs_{*}$. For every $\Ge>0$ we define
  $I_{\Ge}(t_{0})=\{t\ge 0:|t-t_{0}|<\Ge\}$. Saying that $t_{0}$ is in the
  support of $\Gs_{*}$ is equivalent to $\Gs_{*}(I_{\Ge}(t_{0}))>0$ for all $\Ge>0$.
Then, there are two possibilities. Either
  \begin{enumerate}
  \item[(i)] $\displaystyle\lim_{\Ge\to 0}\Gs_{*}(I_{\Ge}(t_{0}))=0$, or
  \item[(ii)] $\displaystyle\lim_{\Ge\to 0}\Gs_{*}(I_{\Ge}(t_{0}))=\Gs_{0}>0$
  \end{enumerate}
Let us first consider case (i). Then we construct a competitor measure
\[
\Gs_{\Ge}(\Gl)=\Gs_{*}(\Gl)-\Gs_{*}|_{I_{\Ge}(t_{0})}+\Gth\Gs_{*}(I_{\Ge}(t_{0}))\Gd_{t_{0}}(\Gl),\quad\Gth>0,
\]
where instead of the distributed mass of $I_{\Ge}(t_{0})$ we place a single
point mass at $t_{0}$. We then define
\begin{equation}
  \label{feps}
  f_{\Ge}(\Go)=\rho_{*}+\int_{0}^{\infty}\frac{d\Gs_{\Ge}(\Gl)}{\Gl-(\Go+ih)^{2}}.
\end{equation}
Formula (\ref{var}) then implies
\[
\lim_{\Ge\to 0}\frac{\|f_{\rm exp}-f_{\Ge}\|_{L^{2}(0,1)}^{2}-
\|f_{\rm exp}-f_{*}\|_{L^{2}(0,1)}^{2}}{\Gs_{*}(I_{\Ge}(t_{0}))}=
(\Gth-1)C(t_{0}).
\]
If $f_{*}$ is a minimizer, then we must have $(\Gth-1)C(t_{0})\ge 0$ for all
$\Gth> 0$, which implies that $C(t_{0})=0$. 

In the case (ii) we have $\Gs_{*}(\{t_{0}\})=\Gs_{0}>0$. Then, for every
$|\Ge|<\Gs_{0}$ we construct a competitor measure
\[
\Gs_{\Ge}(\Gl)=\Gs_{*}(\Gl)+\Ge\Gd_{t_{0}}(\Gl),\quad|\Ge|<\Gs_{0},
\]
as well as the corresponding $f_{\Ge}$, given by (\ref{feps}). We then compute
\begin{equation}
  \label{sigmavar}
  \lim_{\Ge\to 0}\frac{\|f_{\rm exp}-f_{\Ge}\|_{L^{2}(0,1)}^{2}-
\|f_{\rm exp}-f_{*}\|_{L^{2}(0,1)}^{2}}{\Ge}=C(t_{0}).
\end{equation}
Since in this case $\Ge$ can be both positive and negative we conclude that
$C(t_{0})=0$. 

Hence, we have shown that $C(t_{0})=0$ whenever $t_{0}\in[0,+\infty)$ is in
the support of the spectral measure $\Gs$ of the minimizer $f_{*}$. It remains
to observe that for any $t\in\bb{R}$
\[
C(t)=\int_{0}^{1}\frac{f_{\rm exp}(\Go)-f_{*}(\Go)}{t-(\Go-ih)^{2}}d\Go+
\int_{0}^{1}\frac{\bra{f_{\rm exp}(\Go)}-\bra{f_{*}(\Go)}}{t-(\Go+ih)^{2}}d\Go.
\]
Thus, $C(t)$ is a restriction to the real line of a complex analytic function
on the \nbh\ of the real line in the complex $t$-plane. By assumption, $f_{\rm
  exp}\not=f_{*}$, and therefore $C(t)$ is not identically zero. In
particular, the zeros of $C(t)$ cannot have an accumulation point on the real
line. We can also see that the sequence of zeros of $C(t)$ cannot go to infinity by
considering
\[
B(s)=C\left(\nth{s}\right)=s\int_{0}^{1}\frac{f_{\rm exp}(\Go)-f_{*}(\Go)}{1-s(\Go+ih)^{2}}d\Go+
s\int_{0}^{1}\frac{\bra{f_{\rm exp}(\Go)}-\bra{f_{*}(\Go)}}{1-s(\Go-ih)^{2}}d\Go,
\]
which is analytic in a  \nbh\ of 0, and hence cannot have a sequence of zeros
$s_{n}\to 0$, as $n\to\infty$. We conclude that the support of the spectral
measure of the minimizer $f_{*}$ must be finite:
\[
\Gs_{*}(\Gl)=\sum_{j=1}^{N}\Gs_{j}\Gd_{t_{j}}(\Gl),
\]
and the minimizer must be a rational function. 

Now let us consider the competitor (\ref{compet}) defined by $\rho=\rho_{*}$
and $\Gs(\Gl)=\Gs_{*}+\Ge\Gd_{t_{0}}(\Gl)$, where $\Ge>0$ and
$t_{0}\not\in\{t_{1},\ldots,t_{N}\}$. Formula (\ref{var}) then implies that
\[
\Ge C(t_{0})+\Ge^{2}\|\phi_{0}\|^{2}_{L^{2}}\ge
0,\qquad\phi_{0}(\Go)=\frac{1}{t_{0}-(\Go+ih)^{2}}.
\]
for all sufficiently small $\Ge>0$, which implies that $C(t)\ge 0$ for all $t\ge
0$. The necessity of the stated properties of the Caprini function $C(t)$ is
now established. 

Sufficiency is a direct consequence of formula (\ref{var}),
since we can write
\[
\nu(\Gl)=\Gs(\Gl)-\Gs_{*}(\Gl)=\sum_{j=1}^{N}\GD\Gs_{j}\Gd_{t_{j}}(\Gl)+\Tld{\nu}(\Gl),
\]
where $\Tld{\nu}(\Gl)$ is a positive Radon measure without any point masses at
$\Gl=t_{j}$, $j=1,\ldots,N$. We then compute, via formula (\ref{var}), taking
into account that $C(t)\ge0$ for all $t\ge 0$ and $C(t_{j})=0$, that
\[
\|f_{*}+\phi-f_{\rm exp}\|_{L^{2}}^{2}-\|f_{*}-f_{\rm exp}\|_{L^{2}}^{2}=
\GD\rho\lim_{t\to\infty}tC(t)+\int_{0}^{\infty}C(t)d\Tld{\nu}(t)+\|\phi\|_{L^{2}}^{2}\ge 0,
\]
since the first term on the \rhs\ is either nonnegative, if $\rho_{*}=0$ or zero,
if $\rho_{*}>0$.
\end{proof}
We observe that that if $t_{j}>0$, then we must also have $C'(t_{j})=0$, since
$t=t_{j}$ is a point of local minimum of $C(t)$. If we write formula (\ref{fstar}) in
the form
\[
f_{*}(\Go)=\rho_{*}-\frac{\Gs_{0}}{(\Go+ih)^{2}}+\sum_{j=1}^{N}
\frac{\Gs_{j}}{t_{j}-(\Go+ih)^{2}},\qquad\rho_{*}\ge 0,\ \Gs_{0}\ge 0,\ t_{j}>0,\
\Gs_{j}>0,\ j=1,\ldots,N,
\]
then we have exactly $2(N+1)$ equations for $2(N+1)$ unknowns $\rho_{*}$,
$\Gs_{0}$, $t_{j}$, $\Gs_{j}$, $j=1,\ldots,N$:
\[
\rho_{*}\lim_{t\to\infty}tC(t)=0,\quad\Gs_{0}C(0)=0,\quad C(t_{j})=0,\quad
C'(t_{j})=0,\quad j=1,\ldots,N.
\]
Obviously, these equations do not imply that critical points $t_{j}$ are local
minima of $C(t)$, nor do they enforce the nonnegativity of $C(t)$. Taken
together with their highly nonlinear dependence on $t_{j}$ and an unknown
value of $N$, their practical utility for finding $f_{*}$ is dubious. Instead,
Theorem~\ref{th:Caprini} could be used to verify that a particular
$f_{*}(\Go)$ is the minimizer of (\ref{lsq}).

\section{Worst case error analysis}
\setcounter{equation}{0} 
\label{SECT wcea}

\noindent \textbf{Notation:}  We write $A \lesssim B$, if there exists a constant $c$ such that $A \leq c B$ and likewise the notation $A \gtrsim B$ will be used. If both $A \lesssim B$ and $A \gtrsim B$ are satisfied we will write $A\simeq B$. Throughout the paper all the implicit constants will be independent of $\epsilon$. Let also
\begin{equation} \label{S def}
Sf(\omega): = \overline{f(-\overline{\omega})}.
\end{equation}
In this section we analyze the quantity $\GD_{\Go_{0},h}(\epsilon)$, given by
\eqref{Delta} and answer the two questions posed in Section~\ref{SECT main
  results} about $\GD_{\Go_{0},h}(\epsilon)$ by showing that we can restate the
questions entirely in terms of the difference $f-g$. 

\subsection{Reformulation of the problem} \label{SECT Reformulation}

\noindent To analyze $\GD_{\Go_{0},h}(\epsilon)$ we examine the difference $\phi=f-g$. First observe that $\phi$ also has an integral representation (\ref{intrep}) with a signed measure
$\Gs = \sigma_f - \sigma_g$. Let now $\sigma = \sigma^+ - \sigma^-$ be the unique Hahn decomposition of $\sigma$ as a difference of two mutually orthogonal positive measures $\sigma^{\pm}$. Then we may write $\phi=\phi^+ - \phi^-$, where $\phi^\pm \in \PK_h$ are given by
\begin{equation} \label{phi+- def}
\phi^\pm (\omega) : = \int_0^\infty \frac{d \sigma^\pm(\lambda)}{\lambda - (\omega+ih)^2}.
\end{equation}
Thus, we expect that asymptotically $\GD_{\Go_{0},h}(\epsilon)$ and
\begin{equation} \label{Delta* intermid}
\sup \left\{ \frac{|\phi(\Go_{0})|}{\max \|\sigma^\pm\|_*} : \phi \in \PK_h - \PK_h \quad \text{and} \quad \frac{\|\phi\|_{L^2(0,1)}}{\max \|\sigma^\pm\|_*} \leq \epsilon  \right\},
\end{equation}
must be equivalent. Here we have abbreviated $\max \|\sigma^\pm\|_* := \max
(\|\sigma^+\|_*, \|\sigma^-\|_*)$. The next idea comes from the realization
that the asymptotics of the worst possible error is not very sensitive to
specific norms and spaces. The reason, as we have seen in \cite{grho-gen} for
a similar problem, is that the analytic function delivering the largest error
at $\Go_{0}$ is analytic in a larger half-space $\bb{H}_{2h}$ and is therefore
bounded in a wide variety of norms. Our idea is therefore to prove asymptotic
equivalence of $\GD_{\Go_{0},h}(\epsilon)$ to a quadratic optimization problem
in a Hilbert space, permitting us to express the asymptotics of
$\GD_{\Go_{0},h}(\epsilon)$ in terms of the solution of the integral equation
(\ref{Kup}). 

Let us recall the definition of the Hardy class $H^2(\HH_{h})$
\begin{equation} \label{H^2 space def}
H^2(\HH_{h}) = \left\{ f \ \text{is analytic in} \ \HH_h: \sup_{y>-h} \|f\|_{L^2(\RR+iy)} < \infty \right\}.
\end{equation}
It is well known \cite{koos98} that functions in $H^2$ have $L^2$ boundary
data and that $\|f\|_{H^2(\HH_h)} = \|f\|_{L^2(\RR-ih)}$ defines a norm in
$H^2$.  We describe the relation between the Hardy space $H^2(\HH_{h})$ and
$\PK_{h}-\PK_{h}$ more precisely in the following lemma.
\begin{lemma} \label{LEM H2 subset K-K}
Let $f \in H^2(\HH_h)$ with $Sf=f$ and $\int_0^\infty x |\im f(x-ih)| < \infty$, then $f \in \PK_h - \PK_h$ with

\begin{equation} \label{sigma for H^2}
d \sigma(\lambda) = \frac{1}{\pi} \im f(\sqrt{\lambda}-ih) d \lambda
\end{equation}

\noindent Moreover, $f^\pm \in \PK_{h}$ and

\begin{equation} \label{f+- H2 norm bound}
\max \|\Gs_{f^\pm}\|_{*} \leq  \nth{2\sqrt{\pi}}\|f\|_{H^2(\HH_h)}.  
\end{equation}
\end{lemma}

\begin{proof}
We observe that it is enough to prove the lemma for $h=0$ and then apply it to
functions $f(\Go-ih)\in H^{2}(\bb{H}_{+})$, where $f\in H^{2}(\HH_h)$ and $\Go\in\bb{H}_{+}$.
  
For Hardy functions the following representation formula holds (cf. \cite{koos98} p. 128)
\begin{equation} \label{Hardy rep}
f(\omega) = \frac{1}{\pi} \int_{\RR} \frac{\im f(x)}{x-\omega} d x, \qquad \qquad \omega \in \HH_+.
\end{equation}

\noindent Passing to limits in the symmetry relation $Sf(\omega) = f(\omega)$ as $\im \omega \downarrow 0$, and taking imaginary parts we see that $- \im f(x) = \im f(-x) $. The formula \eqref{Hardy rep} now gives

\begin{equation*}
\pi f(\omega) = \int_{0}^\infty \frac{\im f(x)}{x-\omega} d x +  \int_{0}^\infty \frac{\im f(-x)}{-x-\omega} d x = \int_{0}^\infty \frac{2x \im f(x) dx}{x^2-\omega^2} 
=
 \int_0^\infty \frac{\im f(\sqrt{\lambda}) d \lambda}{\lambda-\omega^2},
\end{equation*}
which implies \eqref{sigma for H^2}. 

Next, consider the functions
\begin{equation*}
f^{\pm}(\omega) = \int_0^\infty \frac{d \sigma^\pm(\lambda)}{\lambda-\omega^2},
\qquad \quad
d \sigma^\pm(\lambda) = \frac{1}{\pi} \left(\im f \right)^\pm(\sqrt{\lambda}) d \lambda,
\end{equation*} 
where $(\im f)^\pm$ denote the positive and negative parts of the real valued function $\im f$. Then $f = f^+ - f^-$ and since $\int_0^\infty x |\im f(x)| dx < \infty$, the measures $\sigma^\pm$ are finite and so $f^\pm \in \PK_0$.

Finally, we prove the inequality (\ref{f+- H2 norm bound}). We compute
\[
\|\Gs^{\pm}\|_{*}=\frac{2}{\pi}\int_{0}^{\infty}\frac{x(\im f)^{\pm}(x)}{1+x^{2}}dx.
\]
Applying the Cauchy-Schwarz inequality we obtain 
\[
\|\Gs^{\pm}\|_{*}\le\nth{\sqrt{\pi}}\|(\im f)^{\pm}\|_{L^{2}(0,+\infty)}
\le\nth{\sqrt{\pi}}\|\im f\|_{L^{2}(0,+\infty)}=\nth{2\sqrt{\pi}}\|f\|_{H^{2}(\HH_+)},
\]
where we have used the symmetry and the fact that the
real part of a Hardy function is the Hilbert transform of its imaginary part
\cite{koos98}, and therefore,
\[
\|f\|_{H^2(\HH_+)}^{2} = 2\|\im f\|^{2}_{L^2(\RR)}=4\|\im f\|^{2}_{L^2(0,+\infty)}.
\]
\end{proof}

In order to complete the transition from $\PK_{h}$ to Hardy spaces we need to
replace the norm $\|\Gs\|_{*}$ in (\ref{Delta* intermid}) with an equivalent
Hilbert space norm. This is accomplished in our next Lemma.
\begin{lemma} \label{LEM h' norms equiv}
Let $h' \in (0,h)$, then for any $f \in \PK_h$ 

\begin{equation} \label{h' norm equivalence}
\|f\|_{h'} := \left\| \frac{f}{\omega + ih} \right\|_{H^2(\HH_{h'})} \simeq \|\sigma\|_*,
\end{equation}

\noindent where the implicit constants depend only on $h-h'$.
\end{lemma}

\begin{proof}
Since $\HH_{h'} \subset \HH_h$, it is clear that the function $f(\omega) / (\omega + ih)$ is analytic in $\HH_{h'}$. Next letting $\delta = h-h'$, using the integral representation \eqref{intrep} for $f$ and Fubini's theorem we compute

\begin{equation*}
\begin{split}
\|f\|_{h'}^2 = \int_\RR \frac{1}{x^2 + \delta^2} \int_0^\infty \int_0^\infty \frac{d \sigma(\lambda) d \sigma(t)}{[\lambda - (x+i \delta)^2] [t - (x-i \delta)^2]} dx 
= \int_0^\infty \int_0^\infty I(\lambda,t) \frac{d \sigma(\lambda)}{\lambda + 1} \frac{d \sigma(t)}{t+1},
\end{split}
\end{equation*} 

\noindent where

\begin{equation*}
I(\lambda, t) = \frac{\pi (\lambda + 1) (t+1)}{\delta (\lambda + 4 \delta^2) (t + 4 \delta^2)} \cdot \frac{(\lambda - t)^2 + 12\delta^2 (\lambda + t) + 96 \delta^4}{(\lambda - t)^2 + 8\delta^2 (\lambda + t) + 16 \delta^4}.
\end{equation*}

\noindent This concludes the proof, since it is clear that the function $I(\lambda, t)$ is bounded above and below by two positive constants depending only on $\delta$. 
\end{proof}

Now we are ready to give the desired Hilbert space reformulation of our
problem. For any $h>0$ we define
\begin{equation} \label{D def}
D_{h}(\epsilon) = \sup \left\{ |f(\Go_{0})| \ : f \in H^2(\HH_{h}) \ \text{and} \ Sf=f, \ \|f\|_{H^2(\HH_{h})} \leq 1, \ \|f\|_{L^2(-1,1)} \leq \epsilon \right\}.
\end{equation}
Notice that for convenience we suppressed the dependence on $\Go_{0}$ and also
replaced interval from $[0,1]$ by a symmetric interval $[-1,1]$, resulting in
an equivalent formulation due to the symmetry $Sf=f$ of the functions in
$\PK_h$.

\begin{theorem}[Equivalence of $\Delta$ and $D$] \label{THM Delta* D equiv}
For any $h' \in (0,h)$
\begin{equation} \label{Delta D equiv}
D_{h}(\epsilon) \lesssim \Delta_{h}(\epsilon) \lesssim D_{h'}(\epsilon),
\end{equation}
as $\epsilon\to 0$, where the implicit constants depend only on $h$ and $h'$.
\end{theorem}
\begin{proof}
We first observe that 
\[
  \GD_{h}(\Ge)=\sup\{|f(\Go_{0})-g(\Go_{0})|:\{f,g\}\subset\PK_{h},\ 
\max\{\|\Gs_{f}\|_{*},\|\Gs_{g}\|_{*}\}=1,\ \|f-g\|_{L^{2}(-1,1)}\le\Ge\}.
\]
To prove the first inequality in \eqref{Delta D equiv}, let $\{f,g\}\subset\PK_{h}$ be such that
\[
\max\{\|\Gs_{f}\|_{*},\|\Gs_{g}\|_{*}\}=1,\qquad\|f-g\|_{L^{2}(-1,1)}\le\Ge.
\]
Let
\[
\phi(\Go)=\frac{i(f(\Go)-g(\Go))}{\Go+ih}.
\]
Then, $S\phi=\phi$. Moreover, by Lemma~\ref{LEM h' norms equiv}, for any $h'\in(0,h)$ we estimate
\[
\|\phi\|_{H^{2}(\HH_{h'})}=\|f-g\|_{h'}\le\|f\|_{h'}+\|g\|_{h'}\lesssim
\|\Gs_{f}\|_{*}+\|\Gs_{g}\|_{*}\le 2.
\]
We conclude that there exists a constant $c>0$, depending only on $h$ and
$h'$, such that $c\phi$ is
admissible for $D_{h'}(\Ge)$. Therefore,
\[
D_{h'}(\Ge)\ge c|\phi(\Go_{0})|=\frac{c|f(\Go_{0})-g(\Go_{0})|}{|\Go_{0}+ih|}
\]
Taking supremum over all such pairs $(f,g)$ we conclude that
\[
\GD_{h}(\Ge)\le CD_{h'}(\Ge)
\]
for some constant $C>0$, that depends on $h$ and $h'$, but not on $\Ge$.

To prove the other inequality, let $\phi\in H^{2}(\bb{H}_{h})$ be admissible for
$D_{h}(\Ge)$. The idea is to construct a pair of functions
$\{f,g\}\subset\PK_{h}$ that are admissible for $\GD_{h}(\Ge)$. Since
$\phi$ might not decay sufficiently fast at infinity to be in
$\PK_{h}-\PK_{h}$ we modify it and define
\[
\psi(\Go)=\frac{\phi(\Go)}{(\Go+ih)^{2}}.
\]
This modification preserves the symmetry ($S\psi=\psi$) and
ensures the required decay, so that Lemma~\ref{LEM H2 subset K-K} is
applicable. So that $\psi^{\pm}\in\PK_{h}$ and  
$\|\Gs_{\psi^{\pm}}\|_{*}\lesssim 1$. Now, let $\psi_{0}(\Go)\in\PK_{h}$
be such that $\|\Gs_{\psi_{0}}\|_{*}=1$. We define
\[
F(\Go)=\psi^{+}(\Go)+\psi_{0}(\Go),\qquad
G(\Go)=\psi^{-}(\Go)+\psi_{0}(\Go).
\]
We observe that there exists a constant $C>0$, such that
\[
1=\|\Gs_{\psi_{0}}\|_{*}\le\|\Gs_{F}\|_{*}\le C,\qquad
1=\|\Gs_{\psi_{0}}\|_{*}\le\|\Gs_{G}\|_{*}\le C.
\]
Thus, the pair $(f,g)$ given by
\[
f(\Go)=\frac{F(\Go)}{M},\quad
g(\Go)=\frac{G(\Go)}{M},\quad
M=\max\{\|\Gs_{F}\|_{*},\|\Gs_{G}\|_{*}\}\ge 1
\]
is admissible for $\GD_{h}(\Ge)$. Thus,
\[
\GD_{h}(\Ge)\ge|f(\Go_{0})-g(\Go_{0})|=\frac{|\phi(\Go_{0})|}{(\Go_{0}^{2}+h^{2})M}\ge
\frac{|\phi(\Go_{0})|}{C}.
\]
Taking supremum over all admissible $\phi$ we obtain the remaining inequality
in (\ref{Delta D equiv}).

\end{proof}


\subsection{The effect of the symmetry constraint} 
\label{SECT Optimal bound}

\noindent \textbf{Notation:} Let $H^2:= H^2(\HH_h)$ and let $(\cdot, \cdot)$ and $\|\cdot\|$ denote the inner product and its induced norm in $H^2$. 

\vspace{.1in}

The goal of this section is to analyze the asymptotics of the quantity
$D_{h}(\epsilon)$, as $\Ge\to 0$. Modulo
symmetry $Sf=f$, this has already been done in
\cite{grho-annulus}. Investigating the effect that symmetry may have on the
asymptotics of $D_{h}(\epsilon)$ means relating it to
\begin{equation} \label{sup no symm}
D_{h}^0(\epsilon) = \sup \left\{ |f(\Go_{0})| \ : f \in H^2 \ \text{and} \ \|f\| \leq 1, \ \|f\|_{L^2(-1,1)} \leq \epsilon \right\}.
\end{equation}
The key feature of (\ref{sup no symm}) is its invariance under
multiplying $f$ by a constant phase factor, which allowed us to replace the
target functional $|f(\Go_{0})|$ by a linear one $\re\,f(\Go_{0})$. Since,
multiplication by non-real factors breaks the symmetry $Sf=f$, this reduction
does not work for $D_{h}(\epsilon)$. Nevertheless, convexity of the
target functional permits us to relate it to linear functionals, if
we observe that
\begin{equation*}
  |f(\Go_{0})| = \max_{|\lambda| = 1} \re( \bra{\lambda} f(\Go_{0})).
\end{equation*}
Interchanging the order of maxima with respect to $\lambda$ and $f$ permits us
to use our solution of (\ref{sup no symm}) from \cite{grho-gen} if we can
eliminate the symmetry constraint. This is indeed possible.  Following the
ideas from the theory of reproducing kernel Hilbert spaces \cite{para16} we
write the Cauchy integral formula as an inner product in $H^{2}$: $f(\Go_{0})
= (f, p_{\Go_{0}})$, where $p_{\Go_{0}}$ is given by \eqref{K and p}.  It is
easy to check that for $f\in H^{2}$, satisfying the symmetry constraint we
have

\begin{equation*}
\re( \overline{\lambda} f(\Go_{0}))  = \re (f, \lambda p_{\Go_{0}}) = \re (f,q_{\Go_{0}.\lambda}), \qquad \qquad
q_{\Go_{0},\lambda} = \frac{\lambda p_{\Go_{0}} + S(\lambda p_{\Go_{0}})}{2}.
\end{equation*}
We can now discard the symmetry constraint. We claim that the maximizer function of the problem 

\begin{equation} \label{sup projected}
D_{\Gl,h}^{0}(\Ge)=\sup \left\{ \re (f, q_{\Go_{0},\lambda}) \ : f \in H^2 \ \text{and} \ \|f\| \leq 1, \ \|f\|_{L^2(-1,1)} \leq \epsilon \right\}
\end{equation}
\emph{automatically} has the required symmetry. Indeed, if $f\in
H^{2}$ solves \eqref{sup projected}, we can decompose it into its symmetric
and antisymmetric parts $f=f_s+f_a$, which are mutually real-orthogonal both in
$H^2$ and $L^2(-1,1)$. In other words, they satisfy
\[
\re(f_{s},f_{a})=\re(f_{s},f_{a})_{L^{2}(-1,1)}=0.
\]
Thus, 
\[
\|f\|^2 = \|f_s\|^2 + \|f_a\|^2\ge\|f_s\|^2,\quad
\|f\|^2_{L^{2}(-1,1)} = \|f_s\|^2_{L^{2}(-1,1)} +
\|f_a\|^2_{L^{2}(-1,1)}\ge\|f_s\|^2_{L^{2}(-1,1)},
\] 
which implies that
\[
\Gk=\max \left\{ \|f_s\|,  \frac{\|f_s\|_{L^2(-1,1)}}{\epsilon}
\right\}\le 1.
\]
Also, by the symmetry of
$q_{\Go_{0},\lambda}$ we find that
\[\re (f, q_{\Go_{0},\lambda}) = \re (f_s, q_{\Go_{0},\lambda}).\]
But then the function $f_s /\Gk $ satisfies the constraints of
\eqref{sup projected} and strictly increases the value of target functional
unless $\Gk=1$, or equivalently, $f_a=0$. Thus, if $f$ is the maximizer, then it has to be symmetric.

According to Theorem~\ref{THM D_h^0} from the next section, the maximizer function $f^{*}_{\Ge}(\Go)$ for
(\ref{sup no symm}) has the property that $f^{*}_{\Ge}(\Go_{0})=D_{h}^{0}(\Ge)>0$. Since
removing the symmetry constraint increases the set of admissible functions we
have an obvious inequality
\begin{equation}
  \label{Dobvious}
  D_{h}(\Ge)\leq f^{*}_{\Ge}(\Go_{0})=D_{h}^{0}(\Ge).
\end{equation}
Our foregoing discussion suggests that the function $v_{\Gl,\Ge}=\Gl f^{*}_{\Ge}$
must be a good candidate for the maximizer in $D_{\Gl,h}^{0}(\Ge)$. Using it
as a test function we get the inequality
\[
D_{\Gl,h}^{0}(\Ge)\ge\re(\Gl f^{*}_{\Ge},q_{\Go_{0},\Gl})=
\frac{f^{*}_{\Ge}(\Go_{0})}{2}+\hf\re(\Gl^{2}(f^{*}_{\Ge},Sp_{\Go_{0}})).
\]
We conclude that
\[
D_{h}(\Ge)=\max_{|\Gl|=1}D_{\Gl,h}^{0}(\Ge)\ge\frac{f^{*}_{\Ge}(\Go_{0})}{2}+\hf|(f^{*}_{\Ge},Sp_{\Go_{0}})|\ge\frac{f^{*}_{\Ge}(\Go_{0})}{2}=\hf D_{h}^{0}(\Ge).
\]
Hence, we have shown that
\begin{equation}
  \label{sym_nosym}
  \hf D_{h}^{0}(\Ge)\le D_{h}(\Ge)\le D_{h}^{0}(\Ge).
\end{equation}

\subsection{Optimal bound for $D_h^0(\epsilon)$}

Let us define

\begin{equation} \label{gamma with D_h^0}
\Gg(\Go_{0},h)=\Gg(h)=\limi_{\Ge\to 0}\frac{\ln D^0_{\omega_0, h}(\Ge)}{\ln\Ge}.
\end{equation}

\noindent Combining Theorem~\ref{THM Delta* D equiv} and inequality \eqref{sym_nosym} we see that $D_h^0(\epsilon) \lesssim \Delta_{h} (\epsilon) \lesssim D_{h'}^0(\epsilon)$ for any $h' \in (0,h)$ with implicit constants depending only on $h$ and $h'$. This in particular implies

\begin{equation} \label{gamma h' and gamma h bounds}
\gamma(\omega_0, h') \leq \limi_{\Ge\to 0}\frac{\ln \GD_{\omega_0, h}(\Ge)}{\ln\Ge} \leq \gamma(\omega_0, h), \qquad \qquad \forall h' \in (0,h)
\end{equation}

\noindent It is clear that continuity of $\gamma(\omega_0, h)$ in $h$ will
imply that $\GD_{\omega_0, h}(\epsilon)$ also has power law exponent
$\gamma(\omega_0, h)$. Let us show that the same conclusion will follow under
continuity of $\gamma(\omega_0, h)$ in $\omega_0$ as well. Indeed, it is
enough to show that 
\begin{equation}
  \label{h2om}
 \gamma(\omega_0, h') \geq \gamma\left(\frac{h}{h'}\omega_0, h\right),
\end{equation}
and combine this with \eqref{gamma h' and gamma h bounds}. To prove inequality
(\ref{h2om}), let $f_{\epsilon, \omega_0, h'}^* (\omega)$ be the maximizer
function for $D^0_{\omega_0, h'}(\Ge)$ (cf. Theorem~\ref{THM D_h^0} below) and
consider the function 
\[
g(z) =\sqrt{\frac{h'}{h}}f^*\left(\frac{h'}{h} z\right). 
\]
Note that $\|g\|_{H^2(\HH_h)} = \|f^*\|_{H^2(\HH_{h'})} = 1$ and $\|g\|_{L^2(-1,1)} \leq \|f^*\|_{L^2(-1,1)} = \epsilon$. Therefore, $g$ is an admissible function for $D^0_{\frac{h\omega_0}{h'}, h'}(\Ge)$, hence   

\begin{equation*}
D^0_{\frac{h\omega_0}{h'}, h'}(\Ge) \geq g(\tfrac{h\omega_0}{h'}) =\sqrt{\frac{h'}{h}} f^*(\omega_0) =\sqrt{\frac{h'}{h}} D^0_{\omega_0, h'}(\Ge),
\end{equation*}
which implies inequality (\ref{h2om}). In particular, inequalities (\ref{gamma
  h' and gamma h bounds}) and (\ref{h2om}) imply that $\Gg(\omega_0, h)$ is a
non-increasing function of $\Go_{0}$.  Numerical computations of
$\gamma(\omega_0, h)$ shown in Figure~\ref{fig:gammaofomega} indicate that
$\gamma(\omega_0, h)$ is indeed a continuous function of $\omega_0$. In
Appendix~\ref{SECT power law bounds} we prove that $\gamma(\omega_0,h)$ is also a
non-decreasing function of $h$, satisfying $\gamma(\Go_{0},h) \in (0,1)$ for any $h>0$ and
that $\lim_{h\to0^+} \gamma(\Go_{0},h) = 0$.

To find $\gamma$ we derive an optimal bound for $D_h^0$. Consider the restriction operator $\CR : H^2(\HH_h) \to L^2(-1,1)$
\cite{parfenov,gps03}, then
$\CK = \CR^* \CR$ is a positive, compact and self-adjoint integral operator defined by \eqref{K and p} (where we suppressed the $h$ dependence from the notation). In particular, $\|f\|_{L^2(-1,1)}^2 = (\CK f, f)$. Multiplying $f$ by a constant phase factor we can rewrite \eqref{sup no symm} as

\begin{equation} \label{maximization problem}
\begin{cases}
\Re (f,p_{\omega_0}) \to \text{max}
\\
(f,f) \leq 1
\\
(\CK f, f) \leq \epsilon^2
\end{cases}
\end{equation}

\begin{theorem} \label{THM D_h^0}
Let $\CK$ and $p_{\omega_0}$ be given by \eqref{K and p} and let $\eta = \eta(\epsilon, h, \omega_0)>0$ be the unique solution of $\|(\CK + \eta)^{-1} p_{\omega_0}\|_{L^2(-1,1)} = \epsilon \|(\CK + \eta)^{-1} p_{\omega_0}\|$, then

\begin{equation} \label{D_h^0 exact maximizer}
D_{h}^{0}(\Ge) = \frac{u^*(\omega_0)}{\|u^*\|}
\end{equation}

\noindent where $u^* = u^*_{\epsilon, h, \omega_0}$ solves the integral equation $(\CK + \eta) u^* = p_{\omega_0}$. In particular, the maximizer function is $f^* = u^* / \|u^*\|$. 

\end{theorem}

We can actually express $D_h^0(\epsilon)$ only in terms of $\eta$.

\begin{lemma} \label{LEM exact formula}
Let $\eta = \eta(\epsilon)>0$ be as in Theorem~\ref{THM D_h^0}, then

\begin{equation} \label{D_h^0 exact exponential}
D_h^0(\epsilon) = C \exp \left\{ -\int_\epsilon^1 \frac{t dt}{t^2 + \eta(t)} \right\}
\end{equation}

\noindent where $C$ is a constant independent of $\epsilon$, namely $C = D_h^0(1)$.
\end{lemma}

\begin{proof}
The definition of $u^*$ implies $u^*(\omega_0) = (u^*, p_{\omega_0}) = (u^*, \CK u^* + \eta u^*) = (u^*, \CK u^*) + \eta (u^*,u^*)$, i.e.

\begin{equation} \label{u(z) and the norms}
u^*(\omega_0) = \|u^*\|_{L^2(-1,1)}^2 + \eta \|u^*\|^2 = (\epsilon^2 + \eta)\|u^*\|^2,
\end{equation}

\noindent where the last step follows from the definition of $\eta$. In particular we find that $D_h^0(\epsilon) = (\epsilon^2 + \eta) \|u^*\|$, therefore it is enough to derive a formula for $\|u^*\|$ in terms of $\eta$. Let us write $u^*_\epsilon$ instead of $u^*$ to show its dependence on $\epsilon$. The key observation is the relation between $\partial_\epsilon u^*_\epsilon(\omega_0)$ and $\|u^*_\epsilon\|$ which we are going to use in \eqref{u(z) and the norms} to deduce the desired formula. Let $\{e_n\}_{n=1}^\infty$ be the orthonormal basis of $H^2$ consisting of the eigenfunctions of $\CK$ with corresponding eigenvalues $\{\lambda_n\}_{n=1}^\infty$. The integral equation for $u^*_\epsilon$ diagonalizes in this basis and we find $(e_n, u^*_\epsilon) = e_n(\omega_0) / (\lambda_n + \eta(\epsilon))$. Therefore,

\begin{equation*}
u^*_\epsilon(\omega_0) = \sum_{n=1}^\infty \frac{|e_n(\omega_0)|^2}{\lambda_n + \eta(\epsilon)}, 
\qquad \qquad
\|u^*_\epsilon\|^2 = \sum_{n=1}^\infty \frac{|e_n(\omega_0)|^2}{\left(\lambda_n + \eta(\epsilon)\right)^2}.
\end{equation*}  

\noindent These formulas readily imply

\begin{equation} \label{u_eps in terms of H2 norm}
\partial_\epsilon u^*_\epsilon(\omega_0) = -\eta'(\epsilon)\|u^*_\epsilon\|^2. 
\end{equation}

\noindent Differentiating \eqref{u(z) and the norms} with respect to $\epsilon$ and using the relation \eqref{u_eps in terms of H2 norm} we find

\begin{equation*}
\left(2\epsilon + \eta'(\epsilon)\right) \|u^*_\epsilon\|^2 + 2  \|u^*_\epsilon\| \left(\epsilon^2 + \eta(\epsilon)\right)  \partial_\epsilon  \|u^*_\epsilon\| =    -\eta'(\epsilon)\|u^*_\epsilon\|^2,
\end{equation*}

\noindent which then gives

\begin{equation} \label{u norm deriv over u norm}
\frac{\partial_\epsilon  \|u^*_\epsilon\|}{\|u^*_\epsilon\|} = - \frac{\epsilon + \eta'(\epsilon)}{\epsilon^2 + \eta(\epsilon)} = - \frac{2\epsilon + \eta'(\epsilon)}{\epsilon^2 + \eta(\epsilon)} + \frac{\epsilon}{\epsilon^2 + \eta(\epsilon)}. 
\end{equation}

\noindent Integrating \eqref{u norm deriv over u norm} we find

\begin{equation}
 \|u^*_\epsilon\| = \frac{C}{\epsilon^2 + \eta(\epsilon)} \exp \left\{ -\int_\epsilon^1 \frac{t dt}{t^2 + \eta(t)} \right\},
\end{equation}

\noindent which concludes the proof.

\end{proof}

Combining \eqref{D_h^0 exact maximizer} with \eqref{u(z) and the norms} on one hand and using \eqref{D_h^0 exact exponential} on the other hand (where we change the variables in the integral), we obtain two different representations for the power law exponent:

\begin{equation} \label{gamma with exact maximizer}
\gamma(h) = \limi_{\epsilon \to 0} \frac{\ln \left( (\epsilon + \tfrac{\eta}{\epsilon}) \|u^*\|_{L^2(-1,1)} \right)}{\ln \epsilon}
= \limi_{t \to +\infty} \frac{1}{t} \int_0^t \frac{dx}{1 + e^{2x}\eta(e^{-x})}  
\end{equation}

\noindent Thus, understanding the asymptotic behavior of $\eta(\epsilon)$ as $\epsilon \to 0$ is crucial for unraveling the above formulas. Expanding the two norms in the eigenbasis of $\CK$, we see that $\eta$ solves

\begin{equation}\label{Phi(eta) def and equation}
\Phi(\eta):= \frac{ \sum_{n=1}^\infty \frac{\lambda_n |e_n(\omega_0)|^2}{\left(\lambda_n + \eta\right)^2} }
{\sum_{n=1}^\infty \frac{|e_n(\omega_0)|^2}{\left(\lambda_n + \eta\right)^2} } = \epsilon^{2}.
\end{equation}

\noindent This equation has a unique solution $\eta=\eta(\Ge) > 0$, because $\Phi(\eta)$ is monotone
increasing (since its derivative can be shown to be positive),
$\Phi(+\infty)=(\CK p_{\omega_0},p_{\omega_0})/\|p_{\omega_0}\|^{2}$ and $\Phi(0^{+})=0$ (see \cite{grho-gen} for technical details). Finding the asymptotics of $\eta(\epsilon)$ lies beyond the capabilities of classical asymptotic methods. Nevertheless, under the purported exponential decay \eqref{exp decay} of eigenvalues and eigenfunctions (at the point $\omega_0$) of $\CK$ we proved in \cite{grho-gen} that $\Phi(\eta) \simeq \eta$ with implicit constants independent of $\eta$, leading to $\eta(\epsilon) \simeq \epsilon^2$ with implicit constants independent of $\epsilon$. Moreover, we also showed that $\|u^*\|_{L^2(-1,1)} \simeq \epsilon^{\frac{2\beta}{\alpha}-1}$, which then implies that the ratio inside the first liminf in \eqref{gamma with exact maximizer} converges as $\epsilon \to 0$ and gives the formula $\gamma(h) = 2\beta / \alpha$.

On the other hand, substituting $\lambda_n, |e_n(\omega_0)|$ in \eqref{Phi(eta) def and equation} with their corresponding exponentials from \eqref{exp decay}, and applying (a version) of Lemma~\ref{LEM series asymp} we can approximate

\begin{equation} \label{Phi(eta) approximation}
\Phi(\eta) \approx \eta L\left(\ln\left(\frac{1}{\eta}\right)\right),
\qquad \qquad
L(\tau) = \frac{ \displaystyle e^{\tau} \sum_{k \in \ZZ} \tfrac{e^{(\alpha+2\beta)k}}{(e^{\alpha k} + e^{-\tau})^2}}{\displaystyle \sum_{k \in \ZZ} \tfrac{e^{2\beta k}}{(e^{\alpha k} + e^{-\tau})^2}}.
\end{equation} 
Note that $L(\tau)$ is an elliptic function with periods $\Ga$ and $2\pi i$, further it has symmetries $\overline{L(\tau)} = L(\overline{\tau})$ and $L(2\beta - \tau) = L(\tau)$. Figure~\ref{FIG L(t)} shows the plot of $L$. Therefore, we expect $\epsilon^{-2} \eta(\epsilon)$ to be oscillatory and periodic as $\epsilon \to 0$, more precisely

\begin{equation*}
\epsilon^{-2} \eta(\epsilon) \sim \nth{L(-2 \ln \epsilon)}.
\end{equation*}

\noindent So the integral averages of the function $r(x) = (1 + e^{2x}\eta(e^{-x}))^{-1}$ in the second formula of \eqref{gamma with exact maximizer} converge to the integral (over one period) of its periodic approximation, namely

\begin{equation*}
\frac{2 \beta}{\alpha} = \gamma(h) = \lim_{t \to +\infty} \frac{1}{t} \int_0^t r(x) dx = \lim_{t \to +\infty} \int_0^1 r(tx) dx = \int_0^1 \frac{L(2x)}{1 + L(2x)} dx 
\end{equation*}

\begin{figure}[h]
\center
\includegraphics[scale=0.3]{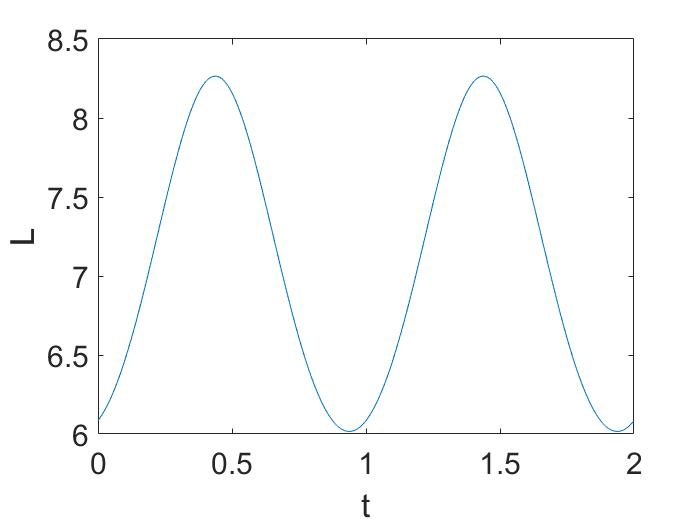}
\caption{The graph of $L(t)$ for $\alpha=4$ and $\beta=1.75$.}
\label{FIG L(t)}
\end{figure}

This insight about the asymptotic behavior of $\eta(\epsilon)$, allowed us to prove a bound that is optimal up to the constant $3/2$, but which is accessible numerically. Namely, with $u=u_{\epsilon, h, \omega_0}$ denoting the solution of the integral equation $(\CK + \epsilon^2) u = p_{\omega_0}$, in \cite{grho-gen} we showed that

\begin{equation*}
D_{h}^{0}(\Ge) \leq \frac{3}{2} u(\omega_0) \min \left\{ \frac{1}{\|u\|}, \frac{\epsilon}{\|u\|_{L^2(-1,1)}} \right\}.
\end{equation*}

\noindent We expect the two quantities under the above minimum to be comparable (this is just a restatement of $\eta(\epsilon) \simeq \epsilon^2$, which holds under \eqref{exp decay}, in fact it also holds under weaker conditions as we observed in \cite{grho-gen}), in which case the formula for $\gamma(h)$ given in \eqref{gamma main formula} follows (compare with the first part of \eqref{gamma with exact maximizer}).
 
The proof of Theorem~\ref{THM D_h^0} follows from \cite{grho-gen} without much change. The only difference is that in the above formulation we presented the exact maximizer for $D_h^0$, versus the $3/2$-maximizer presented in \cite{grho-gen}. For the sake of completeness we give a short recap of the argument.

\begin{proof}[Proof of Theorem~\ref{THM D_h^0}]

\noindent For every $f$, satisfying the two constraints of \eqref{maximization problem} and for every nonnegative numbers $\mu$ and $\nu$ ($\mu^{2}+\nu^{2}\not=0$) we have the inequality
\begin{equation}
  \label{Lagrange}
  ((\mu + \nu \CK)f, f)\le\mu+\nu \epsilon^{2},
\end{equation}

\noindent Applying convex duality to the quadratic functional on the \lhs\ of
(\ref{Lagrange}) we get
\begin{equation}\label{cxdual}
\Re(f, p_{\omega_0})-\frac{1}{2} \left((\mu+\nu \CK)^{-1} p_{\omega_0}, p_{\omega_0} \right)\le
\frac{1}{2} \left( (\mu+\nu \CK)f,f \right)\le
\frac{1}{2} \left( \mu+\nu \epsilon^{2}  \right),
\end{equation}

\noindent so that

\begin{equation} \label{maxub}
\Re(f, p_{\omega_0}) \le \frac{1}{2} \left((\mu+\nu \CK)^{-1} p_{\omega_0}, p_{\omega_0} \right)
+\frac{1}{2} \left( \mu+\nu \epsilon^{2}  \right),
\end{equation}

\noindent which is valid for every $f$, satisfying the constraints of \eqref{maximization problem} and all $\mu > 0$, $\nu\ge
0$. In order for the bound to be optimal we must have equality in
(\ref{cxdual}), which holds if and only if $p_{\omega_0} = (\mu + \nu \CK) f$, giving the formula for optimal vector $f$:

\begin{equation} \label{maxxi}
  f=(\mu + \nu \CK)^{-1} p_{\omega_0}.
\end{equation}
The goal is to choose the Lagrange multipliers $\mu$ and $\nu$ so that the
constraints in (\ref{maximization problem}) are satisfied by $f$, given by (\ref{maxxi}). If $\nu = 0$, then $f = \frac{p_{\omega_0}}{\|p_{\omega_0}\|}$ does not depend on the small parameter $\epsilon$, which leads to a contradiction, because the second constraint $(\CK f,f) \leq \epsilon^2$ is violated when $\epsilon$ is small enough. If $\mu= 0$, then $\CK f = \frac{1}{\nu}p_{\omega_0}$. But this equation has no solution in $H^2$ since $p_{\omega_0}$ has a singularity at $\overline{\omega}_0-2ih$, while $\CK f$ has an analytic extension to $\mathbb{C}\backslash [-1,1]-2ih$.

Thus we are looking for $\mu>0,\  \nu> 0$, so that equalities in
\eqref{maximization problem} hold. (These are the complementary slackness relations in Karush-Kuhn-Tucker conditions.), i.e.

\begin{equation}\label{cxdeq}
\begin{cases}
\left( (\mu + \nu \CK)^{-1} p_{\omega_0}, (\mu + \nu \CK)^{-1} p_{\omega_0} \right) =1, \\
\left( \CK (\mu + \nu \CK)^{-1} p_{\omega_0}, (\mu + \nu \CK)^{-1} p_{\omega_0} \right)= \epsilon^2.
\end{cases}
\end{equation}

\noindent Let $\eta = \frac{\mu}{\nu}$, solving the first equation in \eqref{cxdeq} for $\nu$ we find $\nu = \|(\CK+\eta)^{-1} p_{\omega_0}\|$. The second equation then reads

\begin{equation}\label{Phi(eta) def and equation}
\Phi(\eta):= \frac{ \left( \CK (\CK + \eta)^{-1} p_{\omega_0}, (\CK + \eta)^{-1} p_{\omega_0} \right) }
{ \|(\CK+\eta)^{-1} p_{\omega_0}\|^2 } = \epsilon^{2},
\end{equation}

\noindent which has a unique solution $\eta=\eta(\Ge) > 0$, because $\Phi(\eta)$ is monotone
increasing (since its derivative can be shown to be positive),
$\Phi(+\infty)=(\CK p_{\omega_0},p_{\omega_0})/\|p_{\omega_0}\|^{2}$ and $\Phi(0^{+})=0$ (see \cite{grho-gen} for technical details). Setting $u^*=(\CK+\eta)^{-1}p_{\omega_0}$, \eqref{maxub} reads

\begin{equation} \label{f upper bound intermid}
\Re(f,p_{\omega_0})\le \frac{(u^*,p_{\omega_0})}{2\|u^*\|} + \frac{\|u^*\|}{2} (\epsilon^2 + \eta) = \frac{u^*(\omega_0)}{\|u^*\|},
\end{equation}

\noindent where in the last step we used \eqref{u(z) and the norms}.
\end{proof}

\medskip

\textbf{Acknowledgments.} We are grateful to Leslie Greengard for providing
the quad precision FORTRAN code for solving the integral equation (\ref{Kup})
for $\Ge$ as low as $10^{-16}$ and for computing the eigenvalues of $\CK_{h}$
as small as $10^{-32}$. This material is based upon work supported by the
National Science Foundation under Grant No. DMS-2005538.

\appendix

\section{Appendix} \label{SECT appen}
\setcounter{equation}{0}

\subsection{Extension of positivity}

\begin{proposition} \label{PROP extension}
Let $f$ be analytic in $\HH_h$ with $Sf=f$ (cf. \eqref{S def}) and $f(\omega) \sim -A \omega^{-2}$ as $\omega \to \infty$ for some $A>0$. In addition assume $f'(0) \neq 0$, then the following are equivalent:

\begin{enumerate}
\item[(i)] $\im\, f(x)>0$ for all $x>0$;
\item[(ii)] $\exists\, h' \in (0,h)$ s.t. $\im\, f(x-ih')>0$ for all $x>0$.
\end{enumerate}

\end{proposition}

\begin{proof}
The second item immediately implies the first one. Indeed, the symmetry $Sf=f$ implies that $\im f =0$ on the imaginary axis. Let $\Omega = \{\omega: \im\, \omega > -h', \ \re\, \omega > 0 \}$, note that $\im f \geq 0$ on $\partial \Omega$ and in fact $\min_{\partial \Omega} \im f = 0$, since $\im f$ approaches to zero at infinity (because of $f(\omega) \sim -A \omega^{-2}$) applying the strong maximum principle we conclude that $\im f > 0$ in $\Omega$. (Note that the assumption $f'(0)\neq 0$ was not used here). 

Let us now turn to the converse implication. Let $h_0 \in (0,h)$, then $f$ is
analytic in the closure $\overline{\HH}_{h_0}$ and in particular is bounded
inside the semidisc $D=\{\omega\in\HH_{h_{0}}: |\omega + ih_0| \leq M\}$,
where $M>0$ is a large number that can be chosen such that $|f(\omega)| \leq
2A / |\omega|^2$ for all $\omega \notin D$. With these two inequalities, it is
straightforward to show that $\int_\RR |f(x+iy)|^2 d x$ is bounded uniformly
for $y>-h_0$. Thus, $f \in H^2(\HH_{h_0})$ and following the calculations in
the proof of Lemma~\ref{LEM H2 subset K-K} leading from (\ref{Hardy rep}) to
(\ref{sigma for H^2}), we obtain the representation
$$f(\omega) = \int_0^\infty \frac{d \sigma(\lambda)}{\lambda - (\omega + i h_0)^2}, \qquad \qquad \omega \in \HH_{h_0},$$
where $d \sigma(\lambda) = \frac{1}{\pi} \im f(\sqrt{\lambda}-ih_0) d \lambda$.
Using this, it is easy to find that $f$ must have the more precise asymptotics, as $\omega \to \infty$ in $\HH_{h_0}$:

\begin{equation*}
f(\omega) \sim A \left( -\frac{1}{\omega^2} + \frac{2ih_0}{\omega^3} \right),
\qquad \qquad A = \int_0^{\infty} d \sigma(\lambda).
\end{equation*}

\noindent But then for any $t \in (0,h_0)$

\begin{equation} \label{asymp}
\im\, f(x-it) \sim \frac{2A(h_0-t)}{x^3} > 0, \qquad \quad x \to + \infty.
\end{equation}

Assume, for the sake of contradiction that for each $t \in (0,h_0)$ there
exists $x_t > 0$, such that $\im f(x_t-it) \leq 0$. Clearly, \eqref{asymp}
implies that $x_t$ remains bounded as $t \to 0^+$. Let us now extract
convergent subsequence (without relabeling it) $x_t \to x_0 \geq 0$ as $t \to
0^+$, but then $\im\, f(x_0) \leq 0$. Assumption $(i)$ implies that
$x_0=0$. Let us show that in this case $f'(0)=0$, which is assumed to not be
the case. Since $\im\, f(x_t) > 0$ and $\im f(x_t-it) \leq 0$, by continuity
we conclude that $\exists\, \theta_t \in (0,1]$ such that $\im f(x_t-i \theta_t
t) = 0$. The symmetry $Sf=f$ implies that $\im f(-i\theta_t t) = 0$, therefore
by the mean value theorem $\im f'(\tilde{x}_t - i\theta_t t) = 0$ for some
$\tilde{x}_t \in (0, x_t)$. Taking limits as $t \to 0^+$ we obtain $\im f'(0)
= 0$, but by symmetry $f'(0) \in i\RR$, hence $f'(0)=0$. 
\end{proof}

\subsection{Power law bounds} \label{SECT power law bounds}

Let $D_h^0(\epsilon)$ and $\gamma(h)$ be defined by \eqref{sup no symm} and \eqref{gamma with D_h^0} respectively. Note that $D_h^0(\epsilon)$ is non-increasing in $h$.  Indeed, $\HH_{h_1} \subset \HH_{h_2}$ for $h_1 \leq h_2$ and so admissible functions for $D_{h_2}^0(\epsilon)$ are also admissible for $D_{h_1}^0(\epsilon)$, showing that $D_{h_2}^0(\epsilon) \leq D_{h_1}^0(\epsilon)$. Now dividing by $\ln \epsilon < 0$ and taking liminf in $\epsilon$ we conclude that $\gamma(h)$ is non-decreasing.

Let us turn to deriving power law upper and lower bounds on $D_h^0(\epsilon)$. We are going to use the following two results from \cite{grho-gen} and \cite{grho-annulus}. The first one is analytic continuation from a boundary interval: for any $s \in \HH_+$

\begin{equation} \label{boundary}
\sup\{|f(s)|: f \in H^2(\HH_+) \ \text{and} \ \|f\|_{H^2(\HH_+)} \leq 1, \ \|f\|_{L^2(-1,1)} \leq \delta \} \leq C(s) \delta^{\alpha(s)},
\end{equation}

\noindent where $C(s)^{-2} = \frac{s_i}{9} \left(\arctan \frac{s_r + 1}{s_i} - \arctan \frac{s_r - 1}{s_i}\right)$ with $s = s_r+is_i$ and $\alpha(s) = - \frac{1}{\pi} \arg \frac{s + 1}{s - 1} \in (0,1)$ is the angular size of $[-1,1]$ as seen from $s$, measured in the units of $\pi$ radians. Moreover, the bound is optimal in $\delta$ and maximizer function attaining the bound (up to a constant independent of $\delta$) in \eqref{boundary} is given by 

\begin{equation} \label{G}
G(\zeta) = \frac{\delta}{\zeta - \overline{s}} e^{\frac{i }{\pi} \ln \delta \ln \frac{1+\zeta}{1-\zeta}}, \qquad \qquad \zeta \in \HH_+
\end{equation} 

\noindent where $\ln$ denotes the principal branch of logarithm.

The second one is analytic continuation from a circle. Namely let $\Gamma \subset \HH_+$ be a circle and $s\in\HH_+$ a point lying outside of $\Gamma$, then

\begin{equation} \label{circle}
\sup\{|f(s)|: f \in H^2(\HH_+) \ \text{and} \ \|f\|_{H^2(\HH_+)} \leq 1, \ \|f\|_{L^2(\Gamma)} \leq \epsilon \} \simeq \epsilon^{\beta(s)},
\end{equation}  

\noindent with implicit constants independent of $\epsilon$ and $\beta(s) = \frac{\ln|m(s)|}{\ln\rho}$, where $m$ is the M\"obius map transforming the upper half-plane into the unit disc and the circle $\GG$ into a concentric circle of radius $\rho<1$.

\begin{figure}[h]
\center
\includegraphics[scale=0.3]{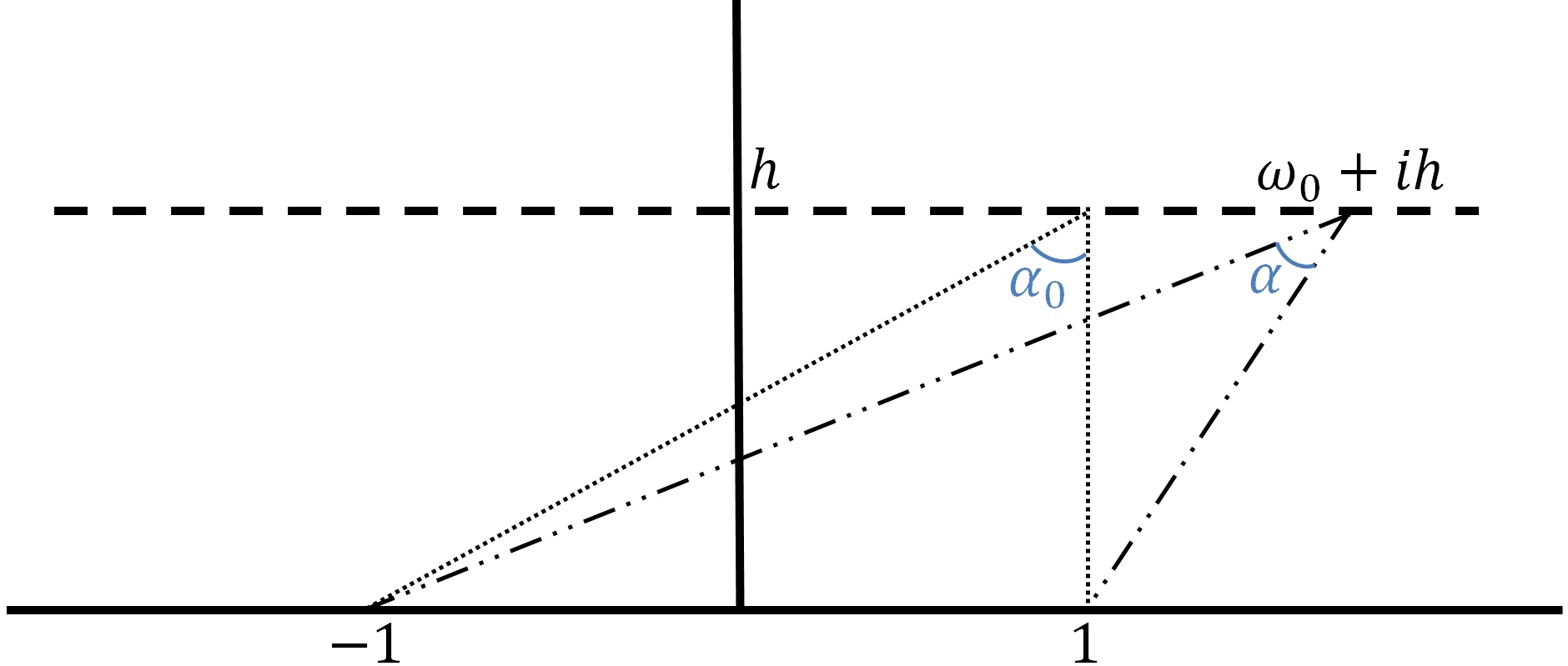}
\caption{Comparison of angles.}
\label{FIG angles}
\end{figure}

\begin{lemma} \label{LEM bounds on D^0}
There exist $\gamma_0, \gamma_1 \in (0,1)$ (depending on $\omega_0, h$) such that

\begin{equation} \label{D power law lower bound}
\epsilon^{\gamma_1} \lesssim D_h^0(\epsilon) \lesssim \epsilon^{\gamma_0},
\end{equation}

\noindent where the implicit constants depend only on $h$ and $\omega_0$. Moreover, $\gamma_1(h) \to 0$ as $h \to 0^+$.
\end{lemma}

\begin{proof}
The lower bound is obtained by introducing an ansatz function admissible for $D_h^0(\epsilon)$. Consider the function $G$ in \eqref{G} with $s=ih$, then the ansatz function is going to be $f(\omega) = G(\omega+ih)$. Note that we can rewrite

\begin{equation*}
G(\zeta) = \frac{\delta^{\alpha(\zeta)} e^{i \theta_\delta(\zeta)}}{\zeta + ih} , \qquad \qquad \theta_\delta(\zeta) = \frac{1}{\pi} \ln \delta \ln \left| \frac{1 + \zeta}{1 - \zeta} \right|.
\end{equation*} 

\noindent It is now clear that

\begin{equation*} 
\|G\|_{L^2((-1,1)+ih)} \lesssim \delta^{\alpha_0}, 
\qquad \qquad
\alpha_0 = \min_{x \in [-1,1]} \alpha(x+ih) = \frac{1}{\pi} \arctan \frac{2}{h} \in (0,1)
\end{equation*}

\noindent and $|G(\omega_0+ih)| \gtrsim \delta^\alpha$, where $\alpha=\alpha(\Go_{0}+ih) < \alpha_0$ (see Figure~\ref{FIG angles}). Thus,

\begin{equation} \label{phi norm ineqs}
\|f\|_{H^2(\HH_h)} \lesssim 1, \qquad\|f\|_{L^2(-1,1)} \lesssim \delta^{\alpha_0}, \qquad
|f(\omega_0)| \gtrsim \delta^\alpha.
\end{equation}

\noindent Letting $\epsilon = \delta^{\alpha_0}$ we see that $cf$ is an admissible function for $D_h^0(\epsilon)$, for some constant $c>0$ independent of $\delta$, hence

\begin{equation*}
D_h^0(\epsilon) \geq c |f(\Go_{0})| \gtrsim \delta^{\alpha} = \epsilon^{\gamma_1},
\end{equation*}

\noindent where $\gamma_1 = \gamma_1(h) = \alpha / \alpha_0 \in (0,1)$. It remains to notice that $\gamma_1(h) \to 0$ as $h \to 0^+$.

Let us now turn to the upper bound. Let $f$ be an admissible function for $D_h^0(\epsilon)$, it is clear that $f$ is also admissible for \eqref{boundary} with $\delta = \epsilon$. However, applying the estimate in \eqref{boundary} at the point $\omega_0>1$ doesn't give a useful bound, since $\alpha(\omega_0) = 0$. Instead let us apply \eqref{boundary} at the points $s$ lying on the circle $\PC = \{s \in \HH_+: |s-i| = \frac{1}{2}\}$. It is clear that the angle $\alpha(s)$ is the smallest at the top point of the circle, i.e. at $s_0 = \tfrac{3}{2} i$. Moreover, obviously the constant $C(s)$ in \eqref{boundary} is uniformly bounded for all $s \in \PC$. Thus,

\begin{equation*}
|f(s)| \lesssim \epsilon^{\beta_0}, \qquad \forall s \in \PC, \quad \text{where} \quad
\beta_0 = \alpha(s_0) = \frac{1}{\pi} \arctan \frac{12}{5}
\end{equation*} 

\noindent and the implicit constant is independent of $s$ and $\epsilon$. In particular, $\|f\|_{L^2(\PC)} \lesssim \epsilon^{\beta_0}$. Now we can apply \eqref{circle} to the function $f(\cdot - ih)$ at the point $s=\omega_0+ih$ and obtain

\begin{equation} \label{f upper bound using circle}
|f(\omega_0)| \lesssim \epsilon^{\gamma_0}, \qquad \qquad
\gamma_0 = \beta_0 \cdot \beta(\omega_0+ih) = \beta_0 \frac{\ln|m(\omega_0+ih)|}{\ln\rho},
\end{equation}  

\noindent where $m(z) = \frac{z-z_0}{z+z_0}$ with $z_0=\frac{i}{2}\sqrt{4h^2+8h+3}$ and $\rho = 2h+2-\sqrt{4h^2+8h+3}$. Taking supremum over $f$ in \eqref{f upper bound using circle} we conclude the proof of the upper bound.
\end{proof}

As an immediate corollary from Lemma~\ref{LEM bounds on D^0} we see that for any $h>0$

\begin{equation*}
\gamma(h) \in [\gamma_0(h), \gamma_1(h)] \subset (0,1)
\end{equation*}

\noindent and also $\gamma(h) \to 0$ as $h \to 0^+$.

\bibliographystyle{abbrv}
\bibliography{refs}

\end{document}